\documentclass[a4paper,UKenglish,envcountsame,envcountsect]{llncs}

\usepackage{microtype}

\usepackage[utf8]{inputenc}
\RequirePackage[T1]{fontenc}
\RequirePackage{lmodern}
\usepackage{amstext,amssymb,amsmath}
\usepackage{relsize,xspace,ifthen}
\usepackage{times}
\usepackage[inline]{enumitem}
\usepackage{tikz,wrapfig,lipsum}
\usepackage{appendix}

\usetikzlibrary{arrows}
\usetikzlibrary{shapes.geometric}
\usetikzlibrary{calc}
\usetikzlibrary{decorations.markings}
\usepackage{hyperref}
\usepackage{comment}

\newcommand{\st}{\mathrel : }

\newenvironment{cenv}{\begin{list}{}{%
      \setlength{\labelwidth}{1.5em}%
      \setlength{\leftmargin}{\labelwidth}%
      \addtolength{\leftmargin}{\labelsep}%
      \setlength{\listparindent}{0em}%
      \setlength{\topsep}{10pt}%
      \setlength{\itemsep}{5pt}%
      \setlength{\parsep}{0pt}%
    }
  }{
  \end{list}
}

\newcounter{claimcounter}

\newenvironment{Claim}{
  
  \refstepcounter{claimcounter}
  \begin{cenv}
  \item[{Claim \arabic{claimcounter}.}]
  }{
  \end{cenv}
}

\newenvironment{ClaimProof}[1][]{\noindent{%
\ifthenelse{\equal{#1}{}}{{\sl Proof.\ }}{{\sl #1.\ }}%
}}{\hspace*{1em}\nobreak\hfill$\dashv$\endtrivlist\addvspace{2ex plus
0.5ex minus0.1ex}}

\renewenvironment{proof}[1][]
{\setcounter{claimcounter}{0}%
    \ifthenelse{\equal{#1}{}}{\noindent\textit{Proof.
    }}{\noindent\textit{#1. }}%
}%
{\hspace{1em}\nobreak\hfill$\Box$\endtrivlist\addvspace{2ex plus
        0.5ex minus0.1ex}}
\newenvironment{Proof}[1][]
{\setcounter{claimcounter}{0}%
    \ifthenelse{\equal{#1}{}}{\noindent\textit{Proof.
        }}{\noindent\textit{#1. }}%
    }%
    {\hspace{1em}\nobreak\hfill$\Box$\endtrivlist\addvspace{2ex 
    plus
            0.5ex minus0.1ex}}

\usepackage{algorithm}
\usepackage{algpseudocode}
\usepackage{eqparbox}

\newenvironment{apthm}[1]{\par\addvspace{3mm}\noindent\textbf {Theorem~\ref{#1}\;}}{\par\addvspace{3mm}}

\numberwithin{equation}{section}
\numberwithin{figure}{section}

\newcommand{\butterfly}{\preceq_b}
\newcommand{\topo}{\preceq_t}

\renewcommand{\st}{\mathrel : }
\newcommand{\N}{{\mathbb N}}

\newcommand{\CCC}{\mathcal{C}}

 \newcommand{\LLL}{\mathcal{L}}
 
 \newcommand{\PPP}{\mathcal{P}}
\newcommand{\QQQ}{\mathcal{Q}} 
 \newcommand{\TTT}{\mathcal{T}}

\newcommand{\Oof}{\mbox{$\cal O$}}

\usepackage{xargs}                      
\usepackage{xcolor}  
\usepackage[colorinlistoftodos,prependcaption,textsize=tiny]{todonotes}
\newcommandx{\complete}[2][1=]{\todo[linecolor=red,backgroundcolor=red!25,bordercolor=red,#1]{#2}}
\newcommandx{\change}[2][1=]{\todo[linecolor=blue,backgroundcolor=blue!25,bordercolor=blue,#1]{#2}}
\newcommandx{\stephan}[2][1=]{\todo[linecolor=black,backgroundcolor=black!25,bordercolor=black,#1]{#2}}
\newcommandx{\saeed}[2][1=]{\todo[linecolor=green,backgroundcolor=green!25,bordercolor=green,#1]{#2}}

\renewcommand{\phi}{\varphi}
\renewcommand{\epsilon}{\varepsilon}

\newcommand{\dtw}{\textit{dtw}}

\newcommand{\usec}[1]{\par\smallskip\noindent\textbf{#1. }}

\usepackage{authblk}

\begin{document}

\title{The Erd\H os-P\'osa Property for Directed Graphs}
\titlerunning{The Erd\H os-P\'osa Property for Directed Graphs}

\author{Saeed Akhoondian Amiri\inst{1}\thanks{Saeed Amiri's and Stephan
    Kreutzer's research partly
    supported by the European Research Council (ERC) under the European Union’s Horizon
2020 research and innovation programme (grant agreement No 648527). } \and Ken-ichi
  Kawarabayashi\inst{2}\thanks{    Research partly
    supported by Japan Society for the Promotion of Science,
    Grant-in-Aid for Scientific Research,
    by C and C Foundation, by Kayamori Foundation,
    by Inoue Research Award for Young Scientists and
    by JST ERATO Kawarabayashi Project.} \and Stephan Kreutzer\inst{1}
  \and \\ Paul Wollan\inst{3}}

\institute{Technical University Berlin,
    \texttt{\{saeed.amiri,stephan.kreutzer\}@tu-berlin.de}.
  \and 
National Institute of Informatics,Tokyo, 
\texttt{k\_keniti@nii.ac.jp} \and
University of Rome "La Sapienza", \texttt{paul.wollan@di.uniroma1.it}}

\maketitle  


\pagestyle{plain}

\begin{abstract}
   A classical result by  Erd\H{o}s and P\'osa\cite{EP} states that there is a
 function $f\st \N \rightarrow \N$ such that for every $k$, every graph
 $G$ contains $k$ pairwise vertex disjoint cycles or a set $T$ of at
 most $f(k)$ vertices such that $G-T$ is acyclic.  The generalisation
 of this result to directed graphs is known as 
 Younger's conjecture and was proved by Reed, Robertson,
 Seymour and Thomas in 1996.

This so-called
 Erd\H{o}s-P\'osa property can naturally be generalised to 
 arbitrary graphs and digraphs. Robertson and Seymour proved that a
 graph $H$ has the 
 Erd\H{o}s-P\'osa-property if, and only if, $H$ is planar.

 In this paper we study the corresponding problem for digraphs. 
 We obtain a complete characterisation of the class of strongly
 connected digraphs which have the Erd\H{o}s-P\'osa-property (both for 
 topological
 and butterfly minors). We also generalise this result to classes of
 digraphs which are not strongly connected. In particular, we study the class of
 vertex-cyclic digraphs (digraphs without trivial strong components). For this natural class of
 digraphs we obtain a nearly complete characterisation of the
 digraphs within this class with the Erd\H{o}s-P\'osa-property. In
 particular we give positive and algorithmic examples of digraphs with the Erd\H{o}s-P\'osa-property by using directed tree
 decompositions in a novel way.
\end{abstract}


\section{Introduction}
A classical result by  Erd\H{o}s and P\'osa states that there is a
function $f\st \N \rightarrow \N$ such that for every $k$, every graph
$G$ contains $k$ pairwise vertex disjoint cycles or a set $T$ of at
most $f(k)$ vertices such that $G-T$ is acyclic.

In \cite{GMV}, Robertson and Seymour considered a natural generalisation of
this result to arbitrary graphs: a graph $H$ has the Erd\H{o}s-P\'osa property
if there is a function $f\st \N\rightarrow \N$ such that every graph $G$ either
has $k$ disjoint copies of $H$ as a minor or contains a set $T$ of at most
$f(k)$ vertices such that $H$ is not a minor of $G-T$. They showed that a graph
$H$ has the Erd\H{o}s-P\'osa-property in this sense if, and only if, it is
planar. In fact, their proof implies that $H$ has the Erd\H{o}s-P\'osa-property 
for minors if, and only if, there is an integer $h=f(|H|)$ such that a
$(h \times h)$-grid contains $H$ as a minor.  

It follows from the grid theorem (see \cite{GMV}), which says that if 
the tree width of a given 
graph $G$ is at least $f(t)$, then $G$ has a $(t \times t)$-grid as a minor,
that if there is no $(h \times h)$-grid minor, then the tree width of $G$ is at 
most $f(h)$. 
Finally, it is proved in \cite{GMV} that the Erd\H{o}s-P\'osa-property holds 
within any 
graph of 
bounded tree-width. 
These facts imply the above characterization.

In \cite{Younger73}, Younger conjectured the natural generalisation of
Erd\H{o}s and P\'osa's original result to directed graphs and directed
cycles. This conjecture has received considerable 
attention by the research community, and 
it has been open for nearly a quarter of a century. 
Following several partial results, Younger's
conjecture was eventually confirmed by Reed et al. in \cite{ReedRST96}.

In this paper we consider the generalisation of Younger's conjecture to
arbitrary digraphs $H$. Whereas for undirected graphs, there is an agreed 
notion of
minor, for directed graphs there are several competing proposals. Here
we study the Erd\H{o}s-P\'osa property  for two common
notions of directed minors, \emph{topological minors} and \emph{butterfly
	minors}. See Section~\ref{sec:prelims} for details.

\begin{definition}\label{def:ep-property}
  A digraph $H$ has the \emph{Erd\H
    os-P\'osa property} for topological minors if there is a
  function $f \st \N\rightarrow \N$ such that for all $k\geq 0$, every digraph
  $G$ either contains $k$ disjoint subgraphs each containing $H$ as a
  topological minor or there is a set $S\subseteq V(G)$ of at most
  $f(k)$ vertices such that $G - S$ does not contain $H$ as a
  topological minor. The definition for butterfly minors is analogous.
\end{definition}

For both concepts of minors we give a complete characterisation of the strongly connected
digraphs which have the Erd\H{o}s-P\'osa property. It turns out that 
our characterisation is essentially the analogy of the above mentioned 
Robertson-Seymour theorem 
for undirected graphs. We prove that a digraph $H$ has the Erd\H{o}s-P\'osa-property for 
topological minors (butterfly minors),
if, and only if,
there is an integer $h=f(|H|)$ such that
a cylindrical wall (grid) of order $h$  contains $H$ as a 
topological minor (butterfly minor).

Note that if $H$ is a cycle, then this 
is exactly Younger's conjecture. Hence, the first main result of this paper is 
the following (see Section~\ref{sec:prelims} and~\ref{sec:dtw} for details).

\begin{apthm}{thm:ep-sc}
  Let $H$ be a strongly connected digraph.
  $H$ has the Erd\H os-P\'osa
  property for butterfly (topological) minors if, and only if, 
  there is a cylindrical grid (wall) $G_c$, for some constant $c = c(H)$,
  such that $H$ is a butterly (topological) minor of $G_c$.
  
  Furthermore, for every fixed strongly connected digraph $H$ satisfying 
  these conditions and every $k$ there is a polynomial time algorithm which, 
  given a digraph $G$ as input, either computes $k$ disjoint (butterfly or 
  topological) models of $H$ in 
  $G$ or a set $S$ of $\leq h(k)$ vertices such that $G-S$ does not contain a 
  model of $H$.
\end{apthm}

This result is particularly interesting as there is no similar
classification known for undirected graphs in terms of topological minors.

If $H$ is not strongly connected, then our techniques described above
fails. In fact they fail already in the bounded
directed tree width case. Nevertheless, we are able to provide a far reaching characterisation of the 
Erd\H os-P\'osa property for a much larger class of digraphs.
In particular, we study the much more
general class of \emph{vertex-cyclic} digraphs, i.e.~digraphs without trivial
strong components (components consisting of a single vertex only). For this
natural class of digraphs we obtain the following result (see Section~\ref{sec:ep-vc} for details).

\begin{apthm}{thm:ep-vc}
  Let $H$ be a weakly connected vertex-cyclic digraph.
  If $H$ has the Erd\H os-P\'osa property for
  butterfly (topological) minors, then it is  
  ultra-homogeneous with respect to butterfly (topological)
  embeddings, its maximum degree is at most
  $3$ and every strong 
  component of $H$ is a butterly (topological) minor of some 
  cylindrical grid (wall) $G_k$.
\end{apthm}

We also obtain a positive result as an example of a digraph satisfying
the properties in the previous theorem. This theorem is probably the
most challenging result of the paper using directed tree
decompositions algorithmically in a novel way which may be of independent interest.

\begin{apthm}{thm:two-random-cycles}
  Let $H$ be a digraph consisting of two disjoint cycles joined by a
  single edge. There is a function $h:\N\rightarrow\N$ such that
  for every integer $k$ and every graph $G$ either there are $k$
  distinct topological models of $H$ in $G$ or there is a set
  $S\subseteq V(G)$ such that $|S|\le h(|H|+k)$ and $H\not\topo
  G-S$. 

  Furthermore, for every $H$ and $k$ there is a polynomial-time
  algorithm which either finds $k$ distinct
  topological models of $H$ in $G$ or finds a set $S\subseteq G$ of
  vertices of size at
  most $h(|H|+k)$ which hits every topological model of~$H$ in~$G$.
\end{apthm}


\section{Preliminaries}
\label{sec:prelims}

In this section we briefly recall necessary definitions and fix our notation. 
We denote the set of non-negative integers by $\N$. For $n\in \N$
 we write $[n]$ for
the set of integers $\{1, \dots, n\}$. 

We refer the reader to~\cite{BangJensenG10,Diestel05} for basic concepts of 
graph and digraph theory. 
All graphs and digraphs in this
paper are finite without  loops. We denote the vertex set of $G$ by $V(G)$ and 
its 
edge set by $E(G)$. If $G$ is a digraph and $S\subseteq V(G)$ or 
$S\subseteq E(G)$, then $G[S]$ 
denotes the subgraph of $G$ \emph{induced by $S$}. For $S\subseteq V(G)$ we 
write 
$G-S$ for the subgraph of $G$ induced by $V(G)-S$. For 
vertices $v\in V(G)$, edges $e\in E(G)$ or sets $F\subseteq E(G)$ we define 
$G-v, G-e, G-F$ analogously.
 
The \emph{in-degree} $d_G^-(v)$ of a vertex $v$ in a digraph
$G$ is the number of edges with head $v$ in $G$. The \emph{out-degree} 
$d^+_G(v)$
is the number of edges with tail $v$. By \emph{degree} $d_G(v)$ of $v$ we mean
the sum $d^+_G(v) + d^-_G(v)$. We usually drop the index if $G$ is
clear from the context.

A \emph{strong component} (or component) in a digraph $G$ is a maximal strongly connected 
subgraph. The \emph{block graph} of a digraph $G$ is the digraph obtained from 
$G$ by contracting each strong component into a single vertex.
We call a digraph \emph{weakly connected} if its underlying 
undirected 
graph is connected.

\section{Directed Minors, Directed Grids and Directed Tree-Width}
\label{sec:dtw}

\usec{Directed Minors}
In this section we introduce two different kinds of
minors, butterfly minors (see \cite{JohnsonRobSeyTho01}) and topological 
minors, and establish same basic properties needed below.

\begin{definition}[butterfly minor]\label{def:butterfly}
  Let $G$ be a digraph. An edge $e = (u,v)\in E(G)$ is
  \emph{butterfly-contractible} if $e$ is the only outgoing edge of
  $u$ or the only incoming edge of $v$. In this case the graph $G'$
  obtained from $G$ by butterfly-contracting $e$ is the graph with
  vertex set $(V(G) - \{u,v\}) \cup \{x_{u,v}\}$, where $x_{u,v}$ is a
  fresh vertex. The edges of $G'$ are the same as the edges of $G$ except for the edges
  incident with $u$ or $v$. Instead, the new vertex $x_{u,v}$ has the
  same neighbours as $u$ and $v$, eliminating parallel edges. A
  digraph $H$ is a \emph{butterfly-minor} of $G$, denoted $H\butterfly G$, if 
  it 
  can be obtained
  from a subgraph of $G$ by butterfly contraction.
\end{definition}


We aim at an alternative characterisation of 
butterfly minors. Let $H, G$ be digraphs such that $H\butterfly 
G$. Let $G'$ be a subgraph of $G$ such that $H$ can be obtained 
from $G$ by butterfly contraction and let $F\subseteq E(G')$ be 
the set of edges contracted in $G'$ to form $H$. Then we can 
partition $F$ into disjoint sets $F_1, \dots, F_h$ 
such that the edges in each $F_i$ are contracted to form a 
single vertex. Hence, $h = |V(H)|$ and no two edges $e_1, e_2$ 
from different sets $F_i \not= F_j$ share an endpoint. 
 The edges of $G'$ not in any $F_i$ are in 
one-to-one correspondence to the edges of $H$. Hence, we can also 
define butterfly minors by a map $\mu$ which assigns to every 
edge $e \in E(H)$ an edge $e\in E(G)$ and to every $v\in V(H)$ a 
subgraph $\mu(v) \subseteq G$ which is $G'[F_i]$ for some $i$ as 
above. We call this a \emph{butterfly model} or \emph{image} of $H$ in 
$G$. As shown in the following lemma, we can always choose an 
image such that $\mu(v)$ are particularly simple. 

\begin{lemma}\label{lem:tree-like-butterfly}
	Let $H, G$ be digraphs such that $H\butterfly G$. Then there is a function
	$\mu$ which maps $E(H)$ to $E(G)$ and vertices $v\in
	V(H)$ to subgraphs $\mu(v) \subseteq G$ such that
	\begin{itemize}[nosep]
		\item $\mu(u)\cap \mu(v) =\emptyset$ for any $u\not=v\in E(H)$,
		\item for all $e = (u,v)\in E(H)$ the edge $\mu(e)$ has its tail in
		$\mu(u)$ and its head in $\mu(v)$,
		\item for all $v\in V(H)$, $\mu(v)$ is the union of an in-branching $T_i$ 
		and out-branching $T_o$ which only have their roots in common and such 
		that for every $e\in E(H)$, if $v$ is the head of $e$ then the head of 
		$\mu(e)$ is a vertex in $T_i$ and if $v$ is the tail of $v$ then the tail 
		of $\mu(e)$ is in $T_o$. 
	\end{itemize}
    We call a map $\mu$ as above a \emph{tree-like model} of $H$ 
    in $G$. We define $\mu(H) := \bigcup_{f\in E(H)\cup V(H)}\mu(f)$.
\end{lemma}	
\begin{proof}
    Suppose the claim was false. Then there are digraphs $H, G'$ 
    such that $H\butterfly G'$ but $H$ has no tree-like model in 
    $G'$. We call such a pair a counter example. Choose such a 
    pair and fix $H$. Within all $G$ such that $(H, G)$ is a 
    counter example let $G$ be a digraph of minimal order and, 
    subject to this, with a minimal number of edges.
    
    For any model $\mu$ of $H$ in $G$ let us call the complexity 
    of $\mu$ the number of edges that are contracted. Let $\mu$ 
    be an image of $H$ in $G$ of minimal complexity. We prove by 
    induction on the complexity that $\mu$ is tree-like. 
    Clearly, if the complexity is $0$, i.e.~no edges need to be 
    contracted, then $\mu$ is tree-like.
    So suppose the complexity is at least $1$. Let $G' := \mu(H)\subseteq 
    G$ be the minimal subgraph of $G$ containing all of $\mu$. 
    By the choice of $G$, we have $G'=G$.
    Choose an edge $e = (u,v)\in E(\mu(v))$ for some $v\in V(H)$ 
    that is 
    butterfly-contractible in $G$ and let $G^*$ be the digraph 
    obtained from $G$ by contracting $e$. Let $x$ be the 
    new vertex generated by contracting $e$. Then $H\butterfly 
    G^*$ and, as $G^*$ has lower order than $G$, there is a 
    tree-like model $\mu^*$ of $H$ in $G^*$. If $x$ is not in 
    $\bigcup_{v\in V(H)}\mu(v)$, then $\mu^*$ is a  
    model of $H$ in a proper subgraph of $G$, contradicting the 
    choice of $G$. So there is a $z\in V(H)$ such that $x\in 
    \mu^*(z)$. Let $F^* = E(\mu^*(v))$. 
    
    We define a set $F\subseteq E(G)$ as follows. Every edge in 
    $F^*$ is 
    either an edge in $G$ or has $x$ as one endpoint. 
    If $e = (w, x)\in F^*$ then $(w, u) \in E(G)$ or $(w, v)\in 
    E(G)$ (or both). If $(w, u)$ exists, we add it to 
    $F$, otherwise we add $(w, v)$. 
    Similarly, if $(x, w) \in F^*$, for some $w\in 
    V(G^*)$, then at least one of $(u,w)$ or $(v, w)$ is in 
    $E(G)$. If $(v, w)$ exists, we add it to $F$ and 
    otherwise we add $(u,w)$. For all $v'\not= z\in V(H)$ we 
    set $\mu(v') := \mu^*(v')$ and we set $\mu(z) := G[F]$. 
    Finally, for all edges $e\in V(H)$, if $\mu^*(e)$ does not 
    contain $x$ as an endpoint we set $\mu(e) := \mu^*(e)$. If 
    $\mu^*(e) := e'$ with $e' = (w, x)$ then $(w, u)$ or $(w, 
    v)$ exist in $E(G)$. If $(w, u)\in E(G)$, we set $\mu(e) := 
    (w, u)$ and otherwise we set $\mu(e) := (w, v)$. If $e = (x, 
    w)$ we proceed analogously, setting $\mu(e) := (v, w)$ if it 
    exists and otherwise $\mu(e) := (u, w)$. 
    
    We claim that $\mu$ is a tree-like model of $H$ in $G$. 
    Suppose not. We know that $\mu^*$ is a tree-like model. 
    Hence, for every $v'\not= z$, $\mu(v')$ is tree-like. So 
    only $\mu(z)$ may violate the tree-condition. Furthermore, 
    the edges in $\mu^*(z)$ induce a tree-like model, 
    i.e.~$\mu^*(v)$ consists of the union of an in-branching $T_i$ 
    and an out-branching $T_o$ as in the statement of the lemma.
    One of the vertices in $\mu^*(z)$ is the fresh vertex $x$. 
    Suppose first that $x\in V(T_i)\setminus V(T_o)$. 
    If all incoming edges of $x$ in $\mu^*$ have been replaced 
    by edges with head $u$ and the unique out-going edge by an 
    edge with 
    tail $v$, then $\mu(v)$ is tree-like. So at least one 
    incoming edge of $x$ has been replaced by an edge $e_i = (w, 
    v)$ or the unique out-going edge $e^*_o$ of $x$ has been 
    replaced by $(u, w)$, for some $w\in V(G)$. If only the 
    out-going edge has been replaced by $(u, w)$, then $v$ has 
    no incoming and only one out-going edge to $u$, so we can 
    simply delete $v$ from $\mu(z)$ and obtain a  
    model. But this would violate the choice of $G$. Hence, at 
    least one edge $(w, x)$ has been replaced 
    by $(w, v)$. However, if the out-going edge of $x$ has been 
    replaced by an edge $(v, w)$, then we still have a tree-like 
    model. Hence, the only case where $\mu$ is not tree-like is 
    if the out-going edge of $x$ in $\mu^*$ has been replaced by 
    an edge $(u, w)$ and at least one in-coming edge of $x$ has 
    been replaced by $(w', v)$. However, in this case the edge 
    $(u,v)$ would not have been butterfly contractible in $G$ as 
    it would neither be the only out-going edge of $u$ nor the 
    only incoming edge of $v$, contradicting the choice of the 
    edge $(u,v)$. 
    
    The other cases, i.e.~if $x\in V(T_o)\setminus V(T_i)$ or 
    $x$ is the root of $T_i$ and of $T_o$ are similar. This 
    concludes the proof.
\end{proof}

Hence by Lemma~\ref{lem:tree-like-butterfly}, we can from now on
assume that butterfly-models are 
always tree-like as in the previous lemma. We will now define 
the other kind of minors considered in this paper.

\begin{definition}[topological minor]\label{def:topological}
  Let $H$, $G$ be digraphs. $H$ is a \emph{topological minor} of $G$,
  denoted $H \topo G$, if there is a mapping $\mu$ which maps every
  vertex $v\in V(H)$ to a vertex $\mu(v)\in V(G)$ and assigns to every
  edge $e\in E(H)$ a directed path $\mu(e)\subseteq G$ such that
  \begin{enumerate}[nosep]
  \item $\mu(v)\not=\mu(w)$ for all $v\not=w\in V(H)$ and
  \item if $e = (v,w) \in E(H)$ then $\mu(e)$ is a path linking
    $\mu(v)$ to $\mu(w)$ and $\mu(e) \cap \big(\bigcup_{v\in 
    V(H)}
    \mu(v) \cup \bigcup_{e'\not=e\in E(H)} \mu(e')\big) = \{ 
    \mu(v), \mu(w)\}.$
  \end{enumerate}
  We call $\mu$ a \emph{topological model} of $H$ in $G$ and define $\mu(H) := 
  \bigcup_{f\in E(H)\cup V(H)} \mu(f)$. 
\end{definition}
That is, $H$ is a topological minor of $G$ if $H$ is a subdivision of
a subgraph of $G$.
We also need the following result.

\begin{lemma}\label{lem:butterfly-topo}
  Let $H$ be a digraph of maximum degree at most $3$. If $H\butterfly
  G$, for some digraph $G$, then $H\topo G$. 
\end{lemma}
\begin{proof}
  Let $H\butterfly G$. Hence, there is a tree-like model $\mu$ 
  of $H$ in $G$. Clearly, for $v\in V(H)$, we can choose the 
  in-branching $T_i$ and the out-branching $T_o$ comprising 
  $\mu(v)$ so that there are at most $3$ leaves. For, if a leave 
  of $T_o$ is not the tail of an edge $\mu(e)$, for some $e\in 
  E(H)$, then we can delete it from the model, unless it is the 
  only vertex of $T_o$. Similarly, we can delete leaves of $T_i$ 
  which are not the head of any $\mu(e)$, $e\in E(H)$. But this 
  implies that $T_i\cup T_o$ has only at most $3$ leaves and 
  therefore contains only one vertex $v'$ of degree $>2$. We can 
  therefore map $v$ to $v'$ and edges of $H$ to corresponding 
  paths to obtain $H$ as a topological minor of $G$.
\end{proof}

\usec{Directed Tree-Width}
We briefly recall the definition of directed tree width from~\cite{JohnsonRobSeyTho01}.
By an \emph{arborescence} we mean a directed graph $R$ such that 
$R$
has a vertex $r_0$, called the
\emph{root} of $R$, with the property that for every vertex $r 
\in V (R)$ there is a
unique directed path from $r_0$ to $r$. Thus every arborescence 
arises from
a tree by selecting a root and directing all edges away from the 
root. If
$r, r' \in V (R)$ we write $r' > r$ if $r' \not= r$ and there 
exists a directed path in $R$
with initial vertex $r$ and terminal vertex $r'$. If $e \in 
E(R)$ we write $r' > e$ if
either $r' = r$ or $r' > r$, where $r$ is the head of $e$.

Let $G$ be a digraph and let $Z \subseteq V (G)$.
We say that a set $S \subseteq (V(G)-Z)$
is \emph{$Z$-normal} if there is no directed walk in $G-Z$ with 
the first and the last vertex
in $S$ that uses a vertex of $G-(Z \cup S)$. It follows that 
every $Z$-normal set
is the union of the vertex sets of strongly connected components 
of $G-Z$. It is
straightforward to check that a set $S$ is $Z$-normal if, and 
only if, the vertex sets of the
strongly connected components of $G-Z$ can be numbered $S_1, 
S_2,\dots, S_d$ in such a way
that
\begin{enumerate}[nosep]
    \item
    if $1 \leq i < j \leq d$, then no edge of $G$ has head in 
    $S_i$ and tail in $S_j$,
    and
    \item
    either $S = \emptyset$, or $S = S_i \cup S_{i+1} \cup \dots 
    \cup S_j$ for some integers $i, j$ with
    $1 \leq i \leq j \leq d$.
\end{enumerate}
\begin{definition}\label{def:dtw}
    A \emph{directed tree decomposition} of a digraph $G$ is a 
    triple
    $(T,\beta,\gamma)$, where $T$ is an arborescence,  $\beta 
    \st V(T)
    \rightarrow 2^{V(G)}$ and $\gamma \st E(T) \rightarrow 
    2^{V(G)}$ are
    functions such that
    \begin{enumerate}[nosep]
        \item $\{ \beta(t) \st t\in V(T) \}$ is a partition of 
        $V (G)$  and
        \item if $e \in E(T)$, then $\bigcup \{ \beta(t) \st t 
        \in V (T), t > e\}$ is
        $\gamma(e)$-normal.
    \end{enumerate}
    For any $t\in V(T)$ we define $\Gamma(t) := \beta(t) \cup 
    \bigcup \{
    \gamma(e) \st e \sim t\}$, where $e\sim t$ if $e$ is 
    incident with
    $t$.
    
    The \emph{width} of $(T, \beta, \gamma)$ is the least 
    integer $w$
    such that $|\Gamma(t)| \leq w + 1$  for all $t
    \in V(T)$.  The \emph{directed tree width} of $G$ is the 
    least
    integer $w$ such that $G$ has a directed tree decomposition 
    of width $w$.
\end{definition}

The sets $\beta(t)$ are called the \emph{bags} and the sets
$\gamma(e)$ are called the \emph{guards} of the directed tree 
decomposition. If $t\in V(T)$ we write $T_t$ for the subtree of $T$ rooted at 
$t$ (i.e. the subtree containing all vertices $s$ such that the unique path 
from the root of $T$ to $s$ contains $t$) and we define $\beta(T_t) := 
\bigcup_{s\in V(T_t)} \beta(s)$.
It is easy to see that
the directed tree width of a subdigraph of $G$ is at most the 
directed
tree width of $G$.

We close the section on directed tree-width by the following 
lemma, which we need below.

\begin{lemma}\label{lem:scc-dtw}
    Let $\TTT := (T, \beta, \gamma)$ be a directed tree 
    decomposition of
    a digraph $G$ and let $H$ be a strongly connected subgraph 
    of $G$.
    Let $S \subseteq T$ be the subgraph of $T$ induced by 
    $\beta^{-1}(H)
    := \{ t \in V(T) \st \beta(t)\cap V(H) \not=\emptyset\}$ and 
    let
    $U\subseteq T$ be the (inclusion) minimal subtree of $T$ 
    containing
    all of $S$.
    Then $\Gamma(t)\cap V(H) \not=\emptyset$ for every $t\in 
    V(U)$.
\end{lemma}
\begin{Proof}
    Let $S$ and $U$ be as defined in the statement of the
    lemma. Towards a contradiction suppose that there is some 
    $u\in
    V(U)$ such that $\Gamma(u) \cap V(H) = \emptyset$. Clearly,
    $u\not\in V(S)$. By construction of $U$ this implies that 
    there are
    vertices $s, t\in U$ and $v, v'\in V(H)$ with 
    $v\in\beta(s)$, $v'\in
    \beta(t)$ and $s,t$ are in different components of $U-u$. 
    Let $P_1,
    P_2$ be two paths in $H$ with $P_1$ linking $v$ to $v'$ and 
    $P_2$
    linking $v'$ to $v$.
    
    As $T$ is a tree at least one of $s,t$ must be in the 
    subtree of $T$
    rooted at a child of $u$. Let $c$ be this child and assume
    w.l.o.g.~that $s$ is in the subtree of $T$ rooted at $c$. 
    But then
    $P_1\cdot P_2$ is a directed walk starting and ending in
    $\beta(T_c)$ which contains a vertex, namely $v'$, not in
    $\beta(T_c)$. Hence, by the definition of directed
    tree-decompositions, $P_1\cdot P_2\cap 
    \Gamma(u)\not=\emptyset$,
    contradicting the assumption that $\Gamma(u)\cap V(H) = 
    \emptyset$.
\end{Proof}

The following theorem follows from \cite{JohnsonRobSeyTho01}, see 
e.g.~\cite{KreutzerO14} for details.
A \emph{linkage} in a digraph $G$ is a set $\LLL$
of pairwise
internally vertex disjoint directed paths. The 
\emph{order} $|\LLL|$
is the number of paths in $\LLL$. 
Let $\sigma :=\{ (s_1, t_1), \dots, (s_k, t_k)\}$ be a set of $k$ pairs of
vertices in $G$. A \emph{$\sigma$-linkage} is a 
linkage $\LLL := \{
P_1, \dots, P_k \}$ of order $k$ such that $P_i$ links $s_i$ to
$t_i$. 
 
 \begin{theorem}\label{thm:dtw:main-algo}
 	Let $G$ be a  digraph and $\TTT := (T, \beta, \gamma)$ be a directed
 	tree-decomposition of $G$ of width $w$. Let $k \geq 1$ and $\sigma$ be a set 
 	of $k$ pairs of
 	vertices in $G$. Then it can be decided in time
 	$\Oof(|V(G)|)^{\Oof(k+w)}$ whether $G$
 	contains a $\sigma$-linkage.
 \end{theorem}

From this, we obtain the following algorithmic result that will be needed 
later.

\begin{theorem}\label{thm:comp-dtw-top}
    Let $H$ be a fixed digraph. 
    There is an algorithm running in time 
    $|G|^{O(|H|)\cdot w}$ which, given a digraph $G$ of directed tree-width 
    at 
    most $w$ as input, computes a butterfly model (topological model) of $H$ in $G$ or determines 
    that 
    $H\not\butterfly G$. 
\end{theorem}
    
\begin{proof}
	The proof for both minor models is nearly identical. We therefore only consider the 
	more complicated cases of butterfly minors. Let $H$ be given and let $G$ be 
	a digraph. If $H\butterfly G$ then, by Lemma~\ref{lem:tree-like-butterfly}, 
	there is a tree-like model $\mu$ of $H$ in $G$.
	Hence, every edge $e\in E(H)$ is mapped to an edge $\mu(e)\in E(G)$ and 
	every vertex $v\in V(H)$ is mapped to the union $\mu(v)$ of an in- and 
	out-branching $T_i \cup T_o$. Clearly, the branchings can be chosen so that 
	they have at most $d_H(v)$ leaves and therefore they contain at most 
	$d_H(v)$ vertices of degree more than $2$. In total, therefore there are at 
	most $2|E(H)|$ vertices of degree more than $2$ in $\bigcup_{v\in 
	V(H)}\mu(v)$. Hence, any tree-like model of $H$ in $G$ consists of the 
	$2|E(H)|$ endpoints of the edges $\mu(e)$, $e\in E(H)$, of the at most 
	$2|E(H)|$ vertices of degree more than $2$ and a set of directed pairwise 
	disjoint paths connecting them in a suitable way to form a butterfly model. 
	Hence, to determine whether $H\butterfly G$ we can simply iterate over all 
	choices of $4|E(H)|$ vertices as candidates for the endpoints of edges and 
	high degree vertices and then apply the algorithm in 
	Theorem~\ref{thm:dtw:main-algo} to check for suitable disjoint directed 
	paths. Clearly, for any fixed $H$ and fixed value of $w$ this runs in 
	polynomial time.
\end{proof}

\usec{Directed Grids}
A natural dual to directed tree width are cylindrical grids which we
define next.

\begin{definition}[cylindrical grid and wall]\label{def:cyl-grid}
  A \emph{cylindrical grid} of order $k$, for some $k\geq 1$, is a
  digraph $G_k$ consisting of $k$ directed cycles $C_1, \dots, C_k$,
  pairwise vertex disjoint, together with a set of $2k$ pairwise
  vertex disjoint paths $P_1, \dots, P_{2k}$ such that \parsep-10pt
  \begin{itemize}[nosep]
  \item each path $P_i$ has exactly one vertex in common with each
    cycle $C_j$,
  \item the paths $P_1, \dots, P_{2k}$ appear on each $C_i$ in this
    order
  \item for odd $i$ the cycles $C_1, \dots, C_k$ occur on all $P_i$
    in this order and for even $i$ they occur in reverse order $C_k,
    \dots, C_1$.
  \end{itemize}
For $1\leq i\leq k$ and $1\leq j\leq 2k$ let $x_{i,j}$ be the common vertex of  
$P_j$ and $C_i$. 

  A \emph{cylindrical wall} of order $k$ is the digraph $W_k$ obtained
  from the cylindrical grid $G_k$ of order $k$ by splitting every
  vertex $v$ of total degree $4$ as follows: we replace $v$ by two
  fresh vertices $v_t, v_h$ plus an edge $(v_t, v_h)$ so that every
  edge $(w, v)\in E(G_k)$ is replaced by an edge $(w, v_t)$ and every
  edge $(v, w)\in E(G_k)$ is replaced by an edge $(v_h, w)$. 
\end{definition}

We will also need the following result. The second part follows using 
Lemma~\ref{lem:butterfly-topo}.

\begin{theorem}[\cite{KawarabayashiK15}]\label{thm:grid}
  There is a function $f\st\N\rightarrow \N$
  such that every digraph of directed tree width at least $f(k)$
  contains a cylindrical grid of order $k$ as a butterfly minor and a 
  cylindrical 
  wall $W_k$ as topological minor.
\end{theorem}

Finally, we need the following acyclic variant of a cylindrical 
grid.

\begin{definition}[acyclic grid]\label{def:grid}
  An \emph{acyclic grid} of order $k$ is a pair $(\PPP, \QQQ)$ of sets
  $\mathcal{P}=\{P_1,$ $\dots,P_k\}$,
  $\mathcal{Q}=\{Q_1,\dots,Q_k\}$ of pairwise vertex disjoint paths such that
  \begin{enumerate}[nosep]
  \item for $1 \leq i \leq k$ and $1 \leq j \leq k$, $P_i \cap Q_j$ is
    a single vertex $v_{ij}$,
  \item for $1 \leq i \leq k$, the vertices $v_{i1},\dots,v_{ik}$ are in
    order in $P_i$, and
  \item for $1 \leq j \leq k$, the vertices $v_{1j},\dots,v_{kj}$ are in
    order in $Q_j$.
  \end{enumerate}
\end{definition}

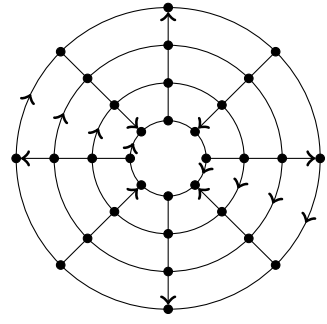
\begin{wrapfigure}{R}{5cm}
	\begin{center}
		\vspace*{-1cm}
		\begin{tikzpicture}[scale=0.5]
		\tikzstyle{vertex}=[shape=circle, fill=black, draw, inner
		sep=.4mm] 
		
		\tikzstyle{emptyvertex}=[shape=circle, fill=white,
		draw, inner sep=.7mm] 
		
		\foreach \x in {1,2,3,4} 
		{ 
			\draw (0,0)        circle (\x); 
			
			\ifnum\x<4 
			\draw[->, very thick] (180-22.5:\x) --
			(180-23.5:\x); \draw[->, very thick] (-22.5:\x) -- (-23.5:\x);
			
			\fi \foreach \z in {0,45,...,350} 
			{ 
				\node[vertex] (n\z\x) at (\z:\x){}; 
			} 
		} 
		\draw[->, very thick] (168-12.25:4) --
		(167-12.25:4); 
		
		\draw[->, very thick] (168+180-12.25:4) --
		(167+180-12.25:4); 
		
		\foreach \z in {0,90,...,350} 
		{ 
			\draw[
			decoration={markings,mark=at position 1 with {\arrow[very
					thick]{>}}}, postaction={decorate} ] (\z:1) -- (n\z4); 
		}
		\foreach \z in {45,135,...,350} 
		{ 
			\draw[
			decoration={markings,mark=at position 1 with {\arrow[very
					thick]{>}}}, postaction={decorate} ] (\z:4) -- (n\z1); 
		}
		\end{tikzpicture}\addtolength{\textfloatsep}{-200pt}
		\caption{Cylindrical grid $G_4$.}\label{fig:grid}
	\end{center}
\end{wrapfigure}

\section{The Erd\H os-P\'osa Property for Strongly Connected Digraphs}

The main result of this section is the following theorem.

\begin{theorem}\label{thm:ep-sc}
   Let $H$ be a strongly connected digraph.
  $H$ has the Erd\H os-P\'osa
  property for butterfly (topological) minors if, and only if, 
  there is a cylindrical grid (wall) $G_c$, for some constant $c = c(H)$,
  such that $H$ is a butterly (topological) minor of $G_c$.
  
  Furthermore, for every fixed strongly connected digraph $H$ satisfying 
  these conditions and every $k$ there is a polynomial time algorithm which, 
  given a digraph $G$ as input, either computes $k$ disjoint (butterfly or 
  topological) models of $H$ in 
  $G$ or a set $S$ of $\leq h(k)$ vertices such that $G-S$ does not contain a 
  model of $H$.
\end{theorem}

We will split the proof of this theorem into two parts. We first show that 
strongly connected digraphs have the Erd\H{o}s-P\'osa-property
within any class of digraphs of bounded directed tree width. Here, a digraph 
$H$ has the Erd\H os-P\'osa-property within a class $\CCC$ of digraphs if the 
condition of Definition~\ref{def:ep-property} is satisfied for every $G\in 
\CCC$. 
\begin{lemma}
\label{lem:bd-dtw}
  Let $\CCC$ be a class of digraphs of bounded directed tree
  width. Then every strongly connected digraph has the Erd\H{o}s-P\'osa-property
  within $\CCC$ with respect to butterfly and topological minors.
\end{lemma}
\begin{Proof}
  We prove the case for butterfly minors, the case for topological minors is 
  almost identical.
    Let $w$ be an upper bound of the directed tree width of all $G\in
  \CCC$. We claim that we can take $f(k) = k\cdot (w+1)$ as function
  witnessing the Erd\H{o}s-P\'osa-property. We provide an algorithm which either
  finds a set $S$ of vertices  of
  size at most $f(k)$ or finds $k$ disjoint copies of $H$ in $G$ as
  a butterfly minor.

  Let $G$ be a digraph such that $\dtw(G)\leq w$ and let $(T, \beta,
  \gamma)$ be a directed tree decomposition of $G$ of width $w$. 
  We prove the claim by induction on $k$. Clearly, for $k=0$ or $k=1$ there 
  is nothing to show. So suppose $k>1$. If $H\not\butterfly G$ then again 
  there is nothing to show. 
  
  So suppose $H\butterfly G$. Let $t\in V(T)$ be a node of minimal height 
  such that $G[\beta(T_t)]$ (see the paragraph following 
  Definition~\ref{def:dtw}) contains $H$ as a butterfly minor.  By the choice 
  of $t$, $G[\beta(T_t)]-\Gamma(t)$ does not 
  contain a model of $H$. Furthermore, by Lemma~\ref{lem:scc-dtw}, no model 
  of $H$ in $G-\Gamma(t)$ can contain a vertex in $\beta(T_t)\setminus 
  \Gamma(t)$ and a vertex of $G-(\beta(T_t)\cup \Gamma(t))$. Hence, all 
  remaining models of $H$ in $G-\Gamma(t)$ must be contained in $G' := 
  G-(\beta(T_t)\cup \Gamma(t))$. By induction hypothesis, either $G'$ 
  contains $k-1$ disjoint models of $H$ as butterfly minor or a set $S$ of 
  $f(k-1)$ vertices such that $G'-S$ does not contain $H$ as a butterfly 
  minor. In the first case we have found $k$ disjoint copies of $H$ as 
  butterfly minor in $G$ and in the second case the set $S' := S\cup \Gamma(t)$ 
  hits every model of $H$. As $|S'|\leq w+1 + f(k-1) \leq k\cdot (w+1) = f(k)$ 
  the 
  claim follows.
\end{Proof}

The next theorem follows from the previous lemma and
Theorem~\ref{thm:grid}.

\begin{theorem}\label{thm:EPPositive}
  Let $H$ be a strongly connected digraph. If there is a $c>0$ such
  that $H \butterfly G_c$ (or $H\topo G_c$), where $G_c$ is the cylindrical 
  grid of order
  $c$, then $H$ has the Erd\H{o}s-P\'osa-property for butterfly minors 
  (resp. topological minors).
  
  Furthermore, for every fixed strongly connected digraph $H$ satisfying 
  these conditions and every $k$ there is a polynomial time algorithm which, 
  given a digraph $G$ as input, either computes $k$ disjoint models of $H$ in 
  $G$ or a set $S$ of $\leq h(k)$ vertices such that $G-S$ does not contain a 
  model of $H$.
\end{theorem}
\begin{proof}
Let $g\st \N\rightarrow \N$ be the function from
  Theorem~\ref{thm:grid} and let $g(k, w) = k\cdot w$ be the
  function as defined in Lemma~\ref{lem:bd-dtw}.
  We claim that the function $h(k) = g(k, f(k\cdot (c+1)))$ witnesses the
  Erd\H{o}s-P\'osa-property for $H$. Towards this aim, let $G$ be a digraph. If 
  $\dtw(G) 
  \geq
  f(k\cdot (c+1))$ then $G$ contains $k$ copies of $G_c$ each of which
  contains $H$ as butterfly minor. Otherwise, $\dtw(G) <
  f(k\cdot (c+1))$ and we can apply Lemma~\ref{lem:bd-dtw}.
  
  Note that, for every fixed $c$, any tree-like butterfly model of $G_c$ in a 
  graph $G$ has directed tree-width 
  bounded by $\Oof(c)$. Hence, we can compute a model of $H$ in 
  any model of $G_c$  in $G$
  by Theorem~\ref{thm:comp-dtw-top}.
\end{proof}

We now show the converse to the previous result.

\begin{theorem}
	\label{thm:EPCounter}
	Every strongly connected digraph $H$ which is not a butterfly minor of some 
	cylindrical grid does not satisfy the Erd\H os-P\'osa property.
\end{theorem}

To prove the theorem we first define a general construction that will be used 
later on. Let $G_k = (C_1, \dots, C_k, P_1, \dots, P_{2k})$ be a cylindrical 
grid, where the $C_i$ are the concentric cycles (ordered from the inside out in 
a fixed embedding of $G_k$ on the plane) and the $P_i$ are the alternating 
paths, ordered in clockwise order on the cycles $C_j$, so that for odd $i$, the 
path $P_i$ traverses the cycles in order $C_1, \dots, C_k$, i.e.~from the 
inside out, whereas for even $i$ the cycles appear on $P_i$ in the reverse 
order.
For $1\leq i\leq k$ and $1\leq j\leq 2k$ let $x_{i,j}$ be the common vertex of  
$P_j$ and $C_i$. 

Recall that a cylindrical wall $W_k$ is obtained from $G_k$ by 
splitting degree $4$ vertices. Note that the outer cycle $C_k$ 
does not have any degree $4$ vertices, and therefore the 
following construction can also be applied to a wall $W_k$.

\begin{definition}[The digraphs $G^{H, e}_n$ and $W^{H, 
e}_n$]\label{def:counter}
	Let $H$ be a digraph and let $e\in E(H)$ be an edge. 
	The digraph $G_k^{H, e}$ is obtained from the disjoint union of $k$ 
	isomorphic copies of $H$, say $H_1, \dots, H_k$, and the grid $G_k$ as 
	follows. In each copy $H_i$ we delete the edge $e_i = (u_i, v_i)$ 
	corresponding to $e$. Furthermore, in $G_k$ we delete all edges 
	$(x_{k,2i-1}, x_{k, 2i})$, for $1\leq i \leq k$. Finally, for all $1\leq i 
	\leq k$, we add an edge $(u_i, x_{k, 2i})$ and an edge $(x_{2i-1}, v_i)$. 
	We call $G_k^{H, e}$ the \emph{attachment of $H$ to $G_k$} and refer to the 
	graphs $H_i$ with the edge $e_i$ deleted plus the two new edges as the 
	\emph{$i$-th copy of $H$} in $G_k^{H, e}$. 
    
    We can apply the same construction using $W_k$ instead of 
    $G_k$. We denote the resulting graph by $W^{H, e}_k$ and 
    call it the \emph{attachment of $H$ to $W_k$}.
\end{definition}

See Figure~\ref{fig:counter} for a schematic overview of 
$G_k^{H, e}$.
We are now ready to prove Theorem~\ref{thm:EPCounter}.

\medskip
\begin{proof}[Proof of Theorem~\ref{thm:EPCounter}]
	Let $H$ be a strongly connected digraph such that $H\not\butterfly G_k$ for 
	all $k\geq 0$. Let $e\in E(H)$. 
	Towards a contradiction, suppose $H$ had the Erd\H os-P\'osa 
	property, witnessed by a function $f\st \N\rightarrow \N$. Choose a value 
	$k>f(2)$	and let $G := G_k^{H, e}$. 

	We first claim that for any set $S\subseteq V(G)$ of at most $f(2)$ 
	vertices, $G-S$ contains $H$ 
	as a butterfly minor. 
	To prove this, let $S$ be such a set. As $|S|<f(2)$, there is an index 
	$1\leq i \leq k$ such that $S$ does not contain a vertex of $C_i\cup 
	P_{2i-1}\cup P_{2i} \cup H_i$, where $H_i$ is the $i$-th copy of $H$ in 
	$G_k^{H,  e}$. But then, $H\butterfly C_i\cup 
	P_{2i-1}\cup P_{2i} \cup H_i$.
	
	To complete the proof we show next that $G$ does not contain two disjoint 
	butterfly models of $H$. Let $\mu$ be a tree-like model of $H$ in $G$. As 
	$H\not\butterfly G_k$, by assumption, $\mu$ must contain a vertex $v$ in 
	some copy $H_i$ of $H$ in $G_k^{H, e}$. But as $H$ is strongly connected and 
	$H_i$ has fewer edges than $H$, $\mu$ must also use both edges $(u_i, x_{k, 
	2i})$ and $(x_{2i-1}, v_i)$ and a directed path in $G_k$ linking $x_{k, 2i}$ 
	to $x_{k, 2i-1}$. We view $G_k$ as being embedded in the plane. 
	Then this path induces a closed curve from $x_{k, 2i}$ 
	to $x_{k, 2i-1}$ in the plane splitting $G_k$ into two disjoint parts. 
	Furthermore, the part containing the rest of the outer cycle $C_k$ not on 
	the curve is acyclic. Hence, there cannot be a second model of $H$ in 
	$G-\mu(H)$. 
\end{proof}
\begin{wrapfigure}{R}{5cm}
	\begin{center}
		\begin{tikzpicture}[scale=0.5]
		\tikzstyle{vertex}=[shape=circle, fill=black, draw, inner
		sep=.4mm] \tikzstyle{emptyvertex}=[shape=circle, 
		fill=white,
		draw, inner sep=.7mm] 
		
		\tikzset{
			partial ellipse/.style args={#1:#2:#3}{
				insert path={+ (#1:#3) arc (#1:#2:#3)}
			}
		}
		
		\foreach \x in {1,2,3,4} { \draw (0,0)
			circle (\x); \ifnum\x<4 \draw[->, very thick] 
			(180-22.5:\x) --
			(180-23.5:\x); \draw[->, very thick] (-22.5:\x) -- 
			(-23.5:\x);
			\fi 
			
			\foreach \z in {0,45,...,350} { \node[vertex] (n\z\x) 
			at
				(\z:\x){}; } } \draw[->, very thick] (168-12.25:4) 
				--
		(167-12.25:4); \draw[->, very thick] (168+180-12.25:4) --
		(167+180-12.25:4); 
		
		\foreach \z in {0,90,...,350} { \draw[
			decoration={markings,mark=at position 1 with 
			{\arrow[very
					thick]{>}}}, postaction={decorate} ] (\z:1) -- 
					(n\z4); }
		
		\foreach \z in {45,135,...,350} { \draw[
			decoration={markings,mark=at position 1 with 
			{\arrow[very
					thick]{>}}}, postaction={decorate} ] (\z:4) -- 
					(n\z1); }
		

		\draw[red,dashed,decoration={markings,mark=at position 1 
		with {\arrow[very
				thick]{>}}}, postaction={decorate}] (3.6,-3.6) -- 
				(5,-0.1);
		
		\draw[decoration={markings,mark=at position 1 with 
		{\arrow[very
				thick]{>}}}, postaction={decorate} ] (4,0) -- 
				(4.9,0);
		
		\node[vertex] at (5,0){};
		
		\node[draw=none] at (4.3,0.2) {$x_1$};
		\node[draw=none] at (5,0.2) {$v_1$};
		
		\node[vertex] at (3.6,-3.6){};
		\node[draw=none] at (3.5,-3.9){$u_1$};
		\node[draw=none] at (2.8,-2.4){$y_1$};
		
		\draw[thick] (5,0) .. controls (8,2) and (5,-7) .. 
		(3.6,-3.6);
		
		\draw[decoration={markings,mark=at position 1 with 
		{\arrow[very
				thick]{>}}}, postaction={decorate}] (3.6,-3.6) -- 
				(2.9,-2.9);
		\node[draw=none] at (6.1,-3){$H'_1$};
		
		\draw[dashed] (0,-5) arc (-90:-135:5);
		\draw[dashed] (-5,0) arc (-180:-225:5);
		\draw[dashed] (0,5) arc (-270:-45-270:5);
		
		\end{tikzpicture}\addtolength{\textfloatsep}{-200pt}
		\caption{Counter example to EP-property for a graph $H$.  
		Just
			$H'_1$ is shown in the figure. Edge $e=(u_1,v_1)$ from 
			$H$ deleted
			and edges $(u_1,y_1)$ and $(x_1,v_1)$ are added to 
			form a
			connection of $H'_1$ to cylindrical grid. 
			}
				\label{fig:counter}
	\end{center}
\end{wrapfigure}
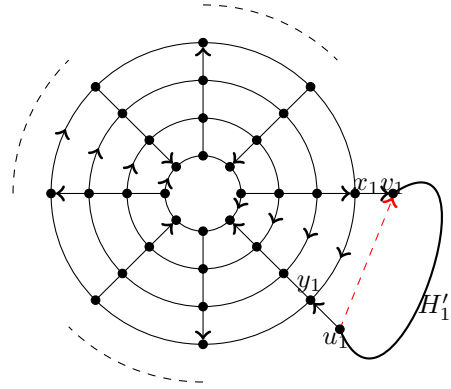
Theorem~\ref{thm:EPCounter} and~\ref{thm:EPPositive} together imply
the proof of Theorem~\ref{thm:ep-sc} for butterfly minors. To prove it
for the topological minors, it is easily seen 
that the same construction as in the Theorem~\ref{thm:EPCounter} where the grid $G_k$ is 
replaced by a wall $W_{k}$ proves that if a strongly connected digraph $H$ is 
not a topological minor of some fixed directed wall $W$, then $H$ does not have 
the  Erd\H os-P\'osa
property for topological minors.

We also obtain the following consequence.
\begin{corollary}
\label{col:EPSCC}
	For strongly connected digraphs, the  Erd\H os-P\'osa
	property (for butterfly and topological minors) is closed under strongly 
	connected subgraphs, i.e.~if a strongly connected graph $H$ does not satisfy 
	the Erd\H os-P\'osa property and $H\butterfly G$ then $G$ does not satisfy
  Erd\H os-P\'osa property.
\end{corollary}

\section{The EP-Property for Vertex-Cyclic Digraphs}
\label{sec:ep-vc}
In this section we extend the results of the previous section to the
more general class of \emph{vertex cyclic} digraphs. A digraph is
\emph{vertex cyclic} if it does not contain a trivial strong
component, i.e.~if every vertex lies on a cycle. Clearly, every
strongly connected digraph is vertex cyclic but the converse is not
true. For simplicity, in this section we only consider weakly connected 
digraphs, i.e.~where 
the underlying undirected graph is connected. Many results can be extended to 
the case of not weakly connected digraphs but we leave this for the full 
version of the paper.

Let $G$ be a digraph and let $e\in E(G)$. Let $n\geq 1$. We define $G^n_e$ as 
the digraph obtained from $G$ by subdividing $e$ $n$ times. 
Given digraphs $H$ and $G$, we say that $H$ is \emph{topologically 
  s-embeddable} in $G$ if there is an edge $e\in E(G)$ such that $H\topo
G^{|H|}_e$. We say that  $H$ is \emph{butterfly s-embeddable} in $G$ if 
there is an edge $e\in E(G)$ such that $H\butterfly G^{|H|}_e$.

  A digraph $G$ is \emph{ultra-homogeneous} with respect to topological (or 
  butterfly) minors, if the block graph of $G$ is a simple directed path 
  without parallel edges and any two components of $G$ are pairwise 
  topologically (or butterfly, respectively) s-embeddable into
  each other and furthermore if the length of the block graph is at
  least $3$, then all of the components except the first and the last
  components, w.r.t. topological order,
  have the same size and also none of those has smaller size than
  the first or the last component.

\begin{definition}\label{def:ultra-homogeneous}
 A digraph $G$ is \emph{ultra-homogeneous} with respect to topological (or 
  butterfly) minors, if the block graph of $G$ is a simple directed path 
  without parallel edges and any two components of $G$ are pairwise 
  topologically (or butterfly, respectively) s-embeddable into
  each other and furthermore if the length of the block graph is at
  least $3$, then all of the components except the first and the last
  components, w.r.t. topological order,
  have the same size and also none of those has smaller size than
  the first or the last component.
\end{definition}

Our main classification result of this section is the following.

\begin{theorem}\label{thm:ep-vc}
  Let $H$ be a weakly connected vertex-cyclic digraph.
  If $H$ has the Erd\H os-P\'osa property for
  butterfly (topological) minors, then it is  
  ultra-homogeneous with respect to butterfly (topological)
  embeddings, its maximum degree is at most
  $3$ and every strong 
  component of $H$ is a butterly (topological) minor of some 
  cylindrical grid (wall) $G_k$.
\end{theorem}

The first result we prove is the following.
\begin{lemma}\label{lem:no-degree-4-vertex}
  Let $H$ be a vertex-cyclic digraph. If $H$ contains a vertex 
  of degree at least
  $4$, then $H$ does not have the Erd\H os-P\'osa property for 
  topological minors.
\end{lemma}
\begin{proof}
  Let $H$ be vertex-cyclic and let $v\in V(H)$ be a vertex of 
  degree 
  at least $4$ in $H$. Furthermore, let $e$ be an incident edge 
  of $v$. 
  Towards a contradiction suppose $H$ had 
  the Erd\H os-P\'osa property witnessed by a function $f\st 
  \N\rightarrow \N$. As in the proof of 
  Theorem~\ref{thm:EPCounter}, let $k>f(2)$. Let $W_k^{H, e}$ be 
  the digraph defined in 
  Definition~\ref{def:counter}. 
  
  We show first that  $W_k^{H, e}$ 
  does not contain two disjoint topological models of $H$. Let 
  $A_1, \dots, A_l$ be the strong components of $H$ in 
  topological order, i.e.~there is no edge from $A_i$ to $A_j$ 
  whenever $j<i$, and let $A_s$ be the component containing $v$. 
  Let $\mu$ be a topological model of $H$ in $W$. Note that in 
  $W_k^{H, e}$ no vertex $w\in V(W_k)$ has degree $\geq 4$. 
  Hence $\mu(v)$ must be in some copy $H'$ of $H$ in $W_k^{H, 
  e}$. More precisely, $\mu(v)$ must be in the strong component 
  of $H'$ corresponding to $A_s$. For, suppose  $\mu(v)$ was in 
  a strong component corresponding to some $A_i$ with $i<s$ such 
  that $A_s$ is reachable from $A_i$ in $H$. Then for every 
  $w\in V(H)$ from which $v$ is reachable in $H$, $\mu(w)$ must 
  be in a component $A_j$ in the copy $H'$ such that $A_i$ is 
  reachable from $A_j$. But this is impossible for cardinality 
  reasons. Similarly, we can show that $\mu(v)$ cannot be 
  in any other component except for $A_s$. It follows that every 
  edge and every vertex of $A_i$ must be mapped to either the 
  copy of $A_i$ in $H'$ or to some vertex of the wall or $A_i$ 
  in another copy of $H$. In any case, $\mu(A_i)$ is 
  strongly connected and therefore $\mu(A_i)$ contains both 
  edges connecting $H'$ to the wall and a directed path between 
  them. We can therefore argue as in the proof of 
  Theorem~\ref{thm:EPCounter}.
  
  To conclude the argument, we can argue as in the proof of 
  Theorem~\ref{thm:EPCounter} that for every set 
  $S\subseteq V(W^{H, e}_k)$ of order $< k$ the graph $W^{H, 
  e}_k-S$ contains $H$ as topological minor.
\end{proof}

Note that this result does not necessarily extend to butterfly minors.
We now introduce a construction that will frequently be applied
below. 
\begin{definition}\label{def:grid-attachment}
	Let $A_k := \big((P_1, \dots, P_k), (Q_1, \dots, Q_k)\big)$ be the 
	acyclic grid of order $k$ as defined in 
	Definition~\ref{def:grid}. Recall that $V(P_i)\cap V(Q_j) = \{v_{i,j}\}$.  
	Let $H$ be a digraph and let $C_i$ and $C_j$ be
	distinct non-trivial strong components of $H$ so that there is
	an edge $e=(u,v)$ for $u\in V(C_i)$ and $v\in V(C_j)$. 
	\begin{enumerate}[nosep]
		\item Let $e_2 \in 
		E(C_2)$ be an edge incident to $v$. 
		The \emph{left 
		acyclic attachment graph $A^{n, H, C_1, C_2}_{e,e_2}$ of $H$ through $e$ 
		and $e_2$ of 
		order $n$} is defined as follows. Take a copy of $A_n
              = (\PPP, \QQQ)$ and 
		$n$ disjoint copies $H_1, \dots, H_n$ of $H$. For every $v\in V(H)$ we 
		write $v^i$ for its isomorphic copy in $H_i$ and likewise we write $e^i$ 
		for the copy of an edge $e$ in $H_i$. For all $1\leq i \leq n$, we delete 
		the edges $e^i = (u^i, v^i)$ and $e_2^i = (x, y)$ and instead add the 
		edge $(u^i, v_{i, 1})$ and identify the topmost vertex $v_{1,i}$ of the 
		$i$-th column of $A_k$ with $x$ and the last vertex $v_{k, i}$ of this 
		column with $y$. 
		\item Now let $e_1\in E(C_1)$ be an edge incident to $u$. The \emph{right 
			acyclic attachment graph $\hat{A}^{n, H, C_1, C_2}_{e,e_1}$ of $H$ 
			through $e$ 
			and $e_1$ of order $n$} is defined analogously, but now the edge $e^i$ 
			in the copy $H_i$ is replaced by an edge $(v_{i, k}, v^i)$ and the end 
			vertices of the edge $e_1 = (x,y)$ are identified with $v_{1, i}$ and 
			$v_{k,i}$ respectively. 
	\end{enumerate}
\end{definition}

It is easily seen that for all $n>1$, $H\topo A^{n, H, C_1, C_2}_{e, e_2}$ and 
hence $H\butterfly A^{n, H, C_1, C_2}_{e, e_2}$, for 
all choices of $C_1, C_2, e, e_2$ as in the definition and 
furthermore, for all $S\subseteq V(A^{n, H, C_1, C_2}_{e, e_2})$ of
order $<n$, $H\topo 
A^{n, H, C_1, C_2}_{e, e_2} - S$ and hence $H\butterfly 
A^{n, H, C_1, C_2}_{e, e_2} - S$.  However, for some 
choices of $H, C_1, C_2$, $A^{n, H}_{e, e_2 C_1, C_2}$ may contain many 
disjoint models of $H$, 
for instance if $H$ only consists of two cycles $C_1, C_2$  connected by an 
edge.

\begin{lemma}\label{lem:biglemma}
Let $H$ be a vertex cyclic digraph and let $\mathcal{C}$ be
the set of its components. If $H$ satisfies any of the following
conditions, then it does not have the Erd\H os-P\'osa property neither
for butterfly nor for topological minors. 
\begin{enumerate}
\item There are $C,C_1,C_2$, all distinct, and edges $e_1,e_2$
	such that $e_l$ links $C$ to $C_l$, for $l=1,2$, or $e_l$
	links $C_l$ to $C$, for $l=1,2$. \label{lem:3components}
\item $H$ contains two components $C$ 
  and $C'$ with two distinct edges linking $C$ to $C'$. \label{lem:2edges}
\item $H$ 
  contains two 
  distinct 
  components $C, C'$ such
  that $C$ is not embeddable into $C'$ (with respect to 
  topological minor). \label{lem:vc-no-nonembed-comp}
   \item  $H$ contains a 
	strong component $C$ such that for all $k\geq 1$, 
	$C\not\butterfly G_k$ (resp. $C\not\topo W_k$). \label{lem:noembed}

\item $H$ has three strongly connected
components $C,C',C''$ such that there is a path from $C$ to $C'$
and a path from $C'$ to $C''$ and either $|C'|<|C''|$ or
$|C'|<|C|$. \label{lem:3cycles}
\end{enumerate}
\end{lemma}
\begin{proof}
We prove the cases for butterfly minors, the cases of topological minors are analogous.
Towards a contradiction, suppose that $H$ has the
Erd\H os-P\'osa property witnessed by a function $f\st
\N\rightarrow\N$.
For each item we construct a counterexample $A$ such that after
deleting $f(2)$ vertices from $A$, it still has a model
of $H$. 
\\\textbf{Proof of Item~\ref{lem:3components}: }
We first consider the case where there is a component 
	$C$ of $H$ and two other components $C_1, C_2$ with an edge from $C$ to 
	$C_1$ and 
	from $C$ to $C_2$. 	
	A terminal component of $H$ is a strong component without any outgoing 
	edges. Let $T$ be a terminal component with a minimal number of edges and 
	among these with a minimal number of vertices. 
	Let $C$ be a component of $H$ with edges to two distinct components and such 
	that $T$ is reachable from $C$ by a path $P$ ($C$ exists as the block graph 
	of $H$ is not a path, by assumption, and $G$ is weakly connected) and let 
	$S$ be the unique 
	component of $H$ such that $P$ contains an edge $e = (s,t)\in E(H)$ with 
	tail $s\in V(S)$ and head $t\in V(T)$. Let $e' = (w, t) 
	\in E(T)$ be an edge with head $t$, which exists as $T$ is not a trivial 
	component. 
	Let $k>f(2)$ and let $A := A^{k, H, S, T}_{e, e'}$ be the left acyclic 
	attachment as defined in Definition~\ref{def:grid-attachment}. See 
	Figure~\ref{fig:3components} for an illustration. 
	
	\begin{wrapfigure}[15]{r}{7cm}
		\begin{center}
			\includegraphics[height=3.7cm]{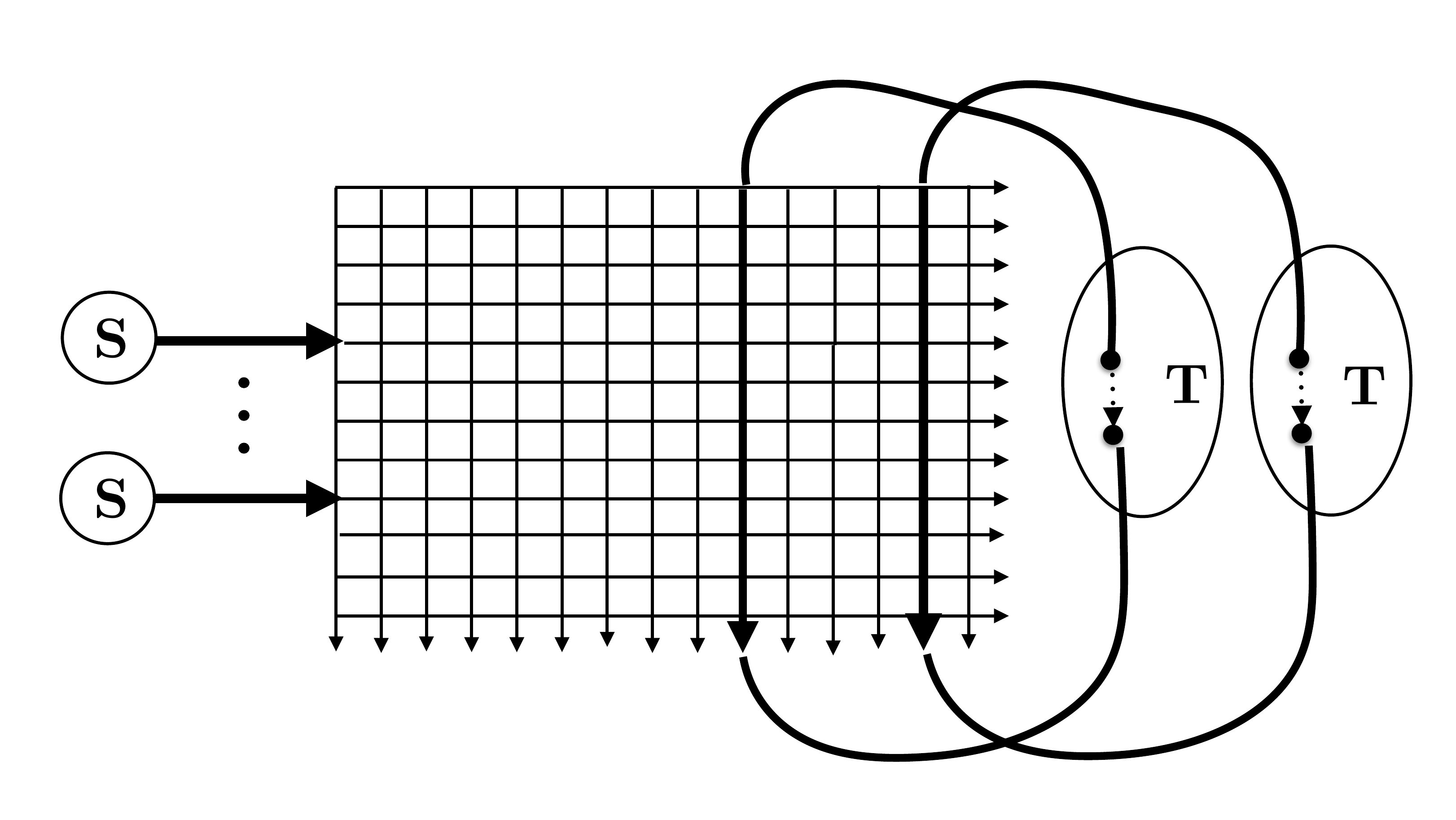}
		\end{center}
		\caption{Illustration of the construction in the proof of 
		Item~\ref{lem:3components} in Lemma~\ref{lem:biglemma}.}\label{fig:3components}
	\end{wrapfigure}
	
	Let $H_1, 
	\dots, H_k$ be the copies of $H$ in $A$. For each vertex $v\in V(H)$ and 
	edge $e\in E(H)$ we write $v^i, e^i$ for the corresponding vertex or edge in 
	$H_i$. Note that $e^i, e'^i$ do not exist as they were deleted in the 
	construction of $A$. We denote by $T^i$ the copy of $T$ in $H_i$ with the 
	edge $e'$ removed and by $\hat{T^i}$ the copy of $T$ in $H_i$ plus the path 
	$Q_i$ of the grid by which $e'$ was replaced.
	
	As noted above, after deleting a set $D$ of $f(2)$ vertices from $A$, 
	$H\butterfly A-D$. Hence it suffices to show that there are no two distinct 
	butterfly models of $H$ in $A$.
	
	Let $\mu$ be a minimal tree-like butterfly model of $H$ in $A$, i.e.~a 
	tree-like model such that no proper subgraph of $\mu(H)$ contains a model of 
	$H$. As $\mu$ is tree-like, every $\mu(v)$ is the union of 
	an in-branching 
	and an out-branching which only share their root $r_v$. 
	As $\mu$ is minimal, if $X\subseteq H$ is strongly 
	connected, then  $\mu(X)$ 
	contains a maximal strongly connected subgraph $\rho(X)$ 
	which contains 
	every root $r_v$ for $v\in V(X)$. It follows that for no 
	component $X$ of 
	$H$ we 
	have 
	$\mu(X)\subseteq T^i$, as $T$ was the component with the minimal number of 
	edges and $T^i$ has one edge less. This implies that if for 
	some vertex $v\in V(H)$, $\mu(v)\cap T^i \not=\emptyset$, for some $1\leq i 
	\leq k$, then 
	$\hat{T^i} \subseteq \mu(X_v)$ where $X_v$ is the component 
	of $H$ 
	containing 
	$v$.
	
	Now consider $\mu(C)$. As $\mu(C)$ is strongly connected, if $\mu(C)$ 
	contains a vertex of the acyclic grid $A_k$ contained in 
	$A$, then $\rho(C) 
	= \hat{T^i}$ for 
	some $1\leq i \leq k$. Let $C_1, C_2$ be two components of $H$ such that $H$ 
	contains an edge $e_1$ from $C$ to $C_1$ and an edge $e_2$ from $C$ to $C_2$.
	But as $T$ was chosen minimal, $\rho(C)\cup 
	\rho(C_l)\not\subseteq 
	\hat{T^i}$, 
	for $l\in \{1,2\}$. Hence, $\rho(C_1) \subseteq 
	\hat{T^{j_1}}$ and 
	$\rho(C_2)\subseteq \hat{T^{j_2}}$ for some $j_1\not=j_2$ 
	different from 
	$i$, 
	as otherwise there was no path from $\mu(C)$ to $\mu(C_1)$ and $\mu(C_2)$ in 
	$A$. But as each of $\mu(C), \mu(C_1), \mu(C_2)$ contains an entire column 
	of the acyclic grid $A_k$ in $A$, this is impossible. 
	
	It follows that $\rho(C)$ must be contained in some 
	$H_i\setminus 
	\hat{T_i}$. 
	As we cannot have $\mu(H)\subseteq H_i\setminus \hat{T_i}$, it 
	follows that for some $j$, $\hat{T_j}\subseteq \mu(H)$ and therefore 
	$\mu(H)$ also includes the edge from $H_i$ to the vertex $x_{i,1}$ of the 
	grid and a path $L_i$ from $x_{i,1}$ to $\hat{T_j}$. 
	
	Now suppose $\mu'$ is a second model of $H$ in $A$, which again we assume to 
	be minimal and tree-like. By the same argument, $\mu'(H)$ must contain an 
	entire column $Q_{j'}$ and path $L_{i'}$ from some vertex $x_{i',1}$ to 
	$Q_{j'}$. But then, if $j'<j$, then $Q_{j'}$ has a non-empty intersection 
	with $i'$ and if $j<j'$ then $Q_j$ has a non-empty intersection with 
	$L_{i'}$. Hence, $\mu$ and $\mu'$ are not disjoint.

	This concludes the case where $H$ contains a component $C$ with two outgoing 
	edges to two distinct other components. The case where there is a component 
	$C$ with incoming edges from two other distinct components is analogous, 
	using the right acyclic attachment instead of the left acyclic
        attachment.
\\\textbf{Proof of Item~\ref{lem:2edges}: } 
Let $C$ and $C'$ be as in the statement of the Item~\ref{lem:2edges} chosen so that from $C'$ no 
	component $X$ of $H$ is reachable such that $X$ has two edges to another 
	component $Y$. Let $e_1 = (s_1, t_1)$ and 
	$e_2 = (s_2, t_2)$ be two distinct edges with tail in $C$ and head in $C'$.

	By Item~\ref{lem:3components} we can assume that the block graph of 
	$H$ is a directed path with parallel edges between components.

	Let $k > f(2)$ and let $A_{2k} = \big( (P_1, \dots, P_{2k}), (Q_1, \dots, 
	Q_{2k})\big)$ be the acyclic grid of order $2k$. Again, $V(P_i)\cap V(Q_j) = 
	\{x_{i,j}\}$.
	Let $G_k$ be the graph obtained from $A_{2k}$ by adding $k$ disjoint copies 
	$H_1, \dots, H_k$ of $H$. For $v\in V(H)$ or $e\in E(H)$ let $v^i$ and $e^i$ 
	be the vertex or edge corresponding to $v$ and $e$ in the copy $H_i$, 
	respectively. For all 
	$1\leq i \leq k$ we delete the edges $e_1^i$ and $e_2^i$ and add edges 
	$(s_1^i, x_{2i-1,1})$, $(s_2^i, x_{2i, 1})$ and $(x_{2k, 2i-1}, t_1^i)$, 
	$(x_{2i, 2k}, t_2^i)$. See Figure~\ref{fig:2edges} for an
        illustration.
	
	\begin{figure}
		\begin{center}
			\includegraphics[height=4cm]{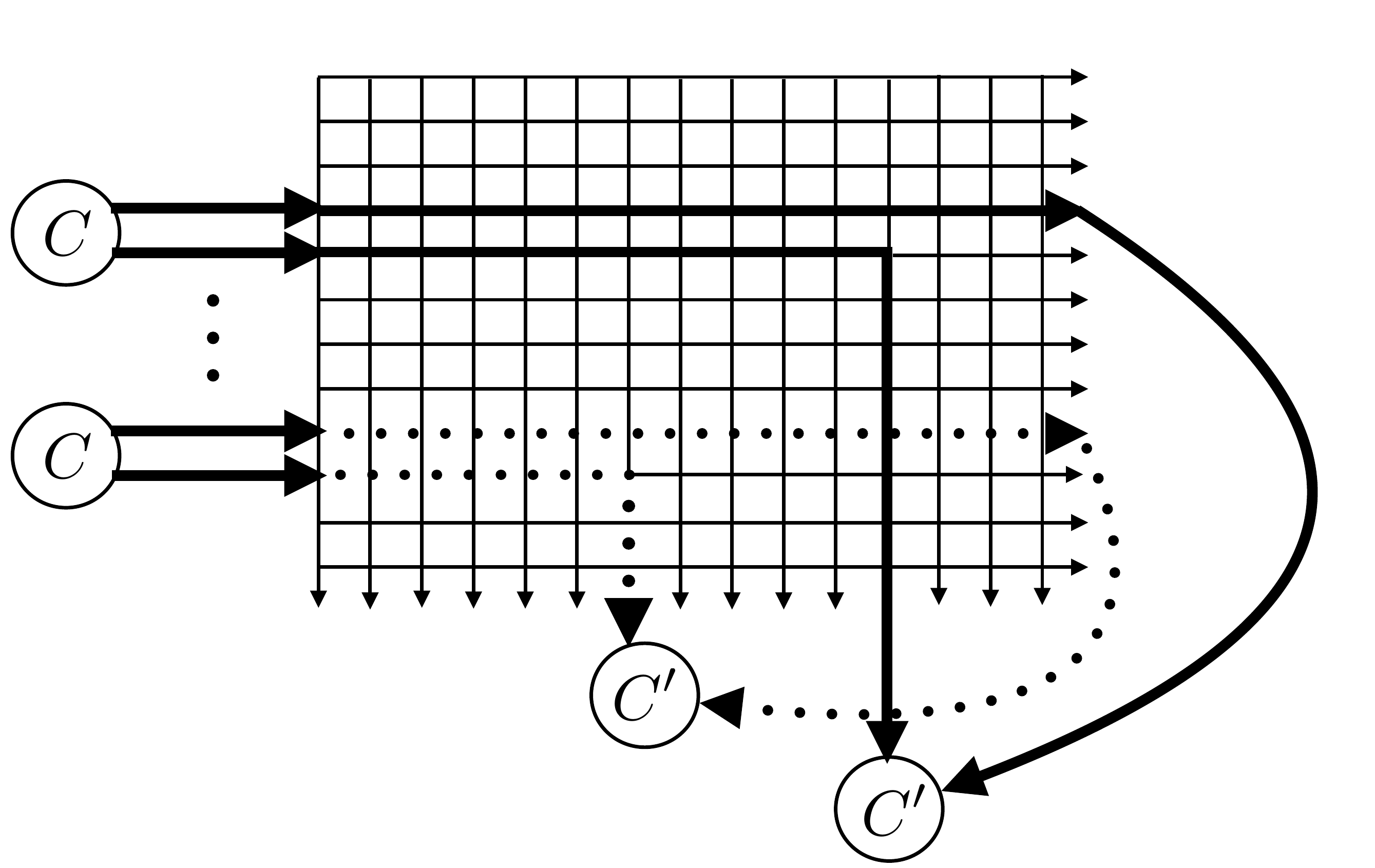}
		\end{center}
		\caption{Illustration of the construction in the proof of 
		Item~\ref{lem:2edges} in Lemma~\ref{lem:biglemma}.}\label{fig:2edges}
	\end{figure} 
	
	As $A_{2k}$ contains two disjoint paths $P^i_1$ linking $x_{2i-1,1}$ to 
	$x_{2k, 
	2i-1}$ and $P^i_2$ linking $x_{2i, 1}$ to $x_{2i, 2k}$, $G_k$ contains $H$ 
	as a butterfly minor. Furthermore, it is easily seen that $H\butterfly G-S$ 
	for every set $S\subseteq G_k$ of order $<k$.
	
	Hence, we only need to show that $G_k$ does not contain two distinct 
	butterfly models of $H$. 
Let $\mu$ be a minimal tree-like butterfly model of $H$ in $A$, i.e.~a 
	tree-like model such that no proper subgraph of $\mu(H)$ contains a model of 
	$H$. As $\mu$ is tree-like, every $\mu(v)$ is the union of 
	an in-branching 
	and an out-branching which only share their root $r_v$. 
	As $\mu$ is minimal, if $X\subseteq H$ is strongly 
	connected, then  $\mu(X)$ 
	contains a maximal strongly connected subgraph $\rho(X)$ 
	which contains 
	every root $r_v$ for $v\in V(X)$. Let $X_1, \dots, 
	X_l$ be the components of $H$ reachable from $C'$ 
	in topological order. By the choice of $C$ and $C'$, between $C'$ and $X_1$ 
	and between $X_i$ and $X_{i+1}$, for all $i<l$, there is exactly one edge. 

	Now, $\mu(C)$ contains a maximal 
	strongly connected subgraph $\rho(C)$ that contains every root $r_v$ for 
	$v\in V(C)$. As $C$ has two outgoing edges, it follows that for all $1\leq i 
	\leq k$, $\rho(C)\cap V(C'^i\cup X_1^i \cup 
	\dots \cup X_l^i)=\emptyset$.  Clearly, $\rho(C)\cap A_{2k} = \emptyset$. 
	Hence, there is an 
	$1\leq i \leq k$ such that $\rho(C)$ is entirely contained in $H_i - 
	V(C'^i\cup X_1^i \cup 
	\dots \cup X_l^i)$. But then, $\mu(H)$ must contain the edges $(s_1^i, 
	x_{2i-1,1})$, $(s_2^i, x_{2i, 1})$ and $(x_{2k, 2i-1}, t_1^i)$, 
	$(x_{2i, 2k}, t_2^i)$ and two disjoint paths $P_1$ linking $x_{2i-1,1}$ to 
	$x_{2k, 
		2i-1}$ and $P_2$ linking $x_{2i, 1}$ to $x_{2i, 2k}$. 
	
	Now let $\mu'$ be another minimal tree-like model of $H$ in $G_k$. By the 
	same argument there must be an index $j$ such that  $\mu'(H)$ contains the 
	edges $(s_1^j, 
	x_{2j-1,1})$, $(s_2^j, x_{2j, 1})$ and $(x_{2k, 2j-1}, t_1^j)$, 
	$(x_{2j, 2k}, t_2^j)$ and two disjoint paths $P'_1$ linking $x_{2j-1,1}$ to 
	$x_{2k, 
		2j-1}$ and $P'_2$ linking $x_{2j, 1}$ to $x_{2j, 2k}$. 
	But clearly, $(P_1\cup P_2) \cap (P'_1\cup P'_2) \not=\emptyset$ and hence 
	the models are not disjoint.
\\\textbf{Proof of Item~\ref{lem:vc-no-nonembed-comp}: } 
Let $H$ and $C, C'$ be as in the statement of the 
	Item~\ref{lem:vc-no-nonembed-comp}.
	By Item~\ref{lem:3components} and~\ref{lem:2edges}, we can 
	assume that the 
	block graph of $H$ is a simple directed path without 
	parallel edges.
	
	Choose $C$ and $C'$ such that $C$ does not embed into $C'$ 
	with respect to 
	butterfly embeddings or vice versa and among all such pairs 
	choose $C'$ so 
	that it is the latest such component in the block graph of 
	$H$, i.e.~no 
	component $C''$ which is part of such a pair is reachable 
	from $C'$. 
	
	We assume that $C$ has no butterfly embedding into $C'$ as 
	defined above. 
	The other case is analogous using right acyclic attachments 
	instead.
	
	Let $S\not=C'$ be the component of $H$ such that $H$ 
	contains an edge $e = 
	(s, 
	t)$ with $s\in V(S)$ and $t\in V(C')$. Let $e' = (w, t)$ be 
	any edge in $C'$ 
	with head $t$, which must exist as $C'$ is not trivial.
	Now let $k>f(2)$ and let $A = A^{k, H, S, C'}_{e, e'}$ be 
	the left acyclic 
	attachment as defined in  
	Definition~\ref{def:grid-attachment}. As before, 
	$H\butterfly A-D$ for any 
	set $D$ of $<k$ vertices. We will show that $H$ has no two 
	disjoint 
	butterfly models in $A$.
	
	Let $\mu$ be a minimal tree-like butterfly model of $H$ in 
	$A$. Let $H_1, 
	\dots, H_k$ be the disjoint copies of $H$ in $A$ and as 
	before we write 
	$v^i, e^i$ for the copy of a vertex $v\in V(H)$ or edge 
	$e\in E(H)$ in the 
	$i$-th copy. Furthermore, as in the previous proofs, as 
	$\mu$ is tree-like 
	and minimal, every $\mu(v)$ is the union of two branchings 
	sharing only 
	their root $r_v$ and for every strongly connected subgraph 
	$X\subseteq H$ 
	the model $\mu(X)$ contains a maximal strongly connected 
	subgraph $\rho(X)$ 
	which contains all roots $r_v$ of $v\in X$
        Let $C'^i$ be 
	the copy of $C'$ 
	in $H_i$ with the edge $e'^i$ removed and let $\hat C'^i$ be 
	the copy of 
	$C'$ in $H_i$ where the edge $e'$ is replaced by the column 
	$Q_i$ of the 
	grid $A_k = \big((P_1, \dots, P_k, Q_1, \dots, Q_k)\big)$ 
	used to construct 
	$A$. 
	As $C$ has no butterfly 
	embedding 
	in $C'$, $\rho(C)$ cannot be contained in $\hat C'^i$ for 
	any $1\leq i \leq 
	k$ and therefore $\rho(C) \subseteq H_i - V(C'^i \cup X^i_1 
	\cup \dots \cup 
	X^i_l)$, for some $1\leq i \leq k$, where $X_1, \dots, X_l$ 
	are the 
	components of $H$ 
	reachable from $C'$.
	On the other hand, $\mu(H)\not\subseteq H_j - V(C'^j \cup 
	X^j_1 \cup \dots 
	\cup 
	X^j_l)$, for any $1\leq j \leq k$. Hence, $\mu(C)$ must 
	contain $\hat C'^j$ 
	for some $j$ and a path $L_i$ from $x_{i, 1}$ to a vertex on 
	$Q_j$, where 
	$x_{i,j}$ is the unique vertex in $V(P_i)\cap V(Q_j)$, for 
	all $1\leq 
	i,j\leq k$.
	
	Now let $\mu'$ be another butterfly model of $H$ in $A$. By 
	the same 
	argument, $\mu'(H)$ must contain a column $Q_{j'}$ and a 
	path $L_{i'}$ from 
	$x_{i', 1}$ to a vertex on $Q_{j'}$. But then $\mu$ and 
	$\mu'$ are not 
	distinct.
\\\textbf{Proof of Item~\ref{lem:noembed}: } A construction very similar to the construction in the 
proof of 
Theorem~\ref{thm:EPCounter} shows that this case holds and we omit the
details.
\\\textbf{Proof of Item~\ref{lem:3cycles}: }
Let $C_1,C_2,C_3$
be a triple of strongly connected components as stated in Item~\ref{lem:3cycles}.
We prove Item~\ref{lem:3cycles} in the case that $|C_2|<|C_3|$ the other
case is analogous. By
items~\ref{lem:3components}, ~\ref{lem:2edges}, we can 
	assume that the 
	block graph of $H$ is a simple directed path without 
	parallel edges.

We can assume that the distance between
$C_1,C_2$ and $C_3$ in the block graph is minimized among all triples
of components satisfying conditions of Item~\ref{lem:3cycles}.

As the distance is minimized one can easily show by a simple case
distinction that $C_1,C_2$ and $C_3$ are three consecutive vertices in
the block graph of $G$.

Let $k>f(2)$. Construct a graph $A^H_k$ as follows. Let $e_1=(u,v)$ be
the edge from $C_1$ to $C_2$ and $e_2=(x,y)$ the edge from $C_2$ to $C_3$. 
Let $H_i=(V(H),E(H)\setminus \{e_1,e_2\})$ for all $i\in [k]$.
	\begin{figure}
		\begin{center}
			\includegraphics[height=4cm]{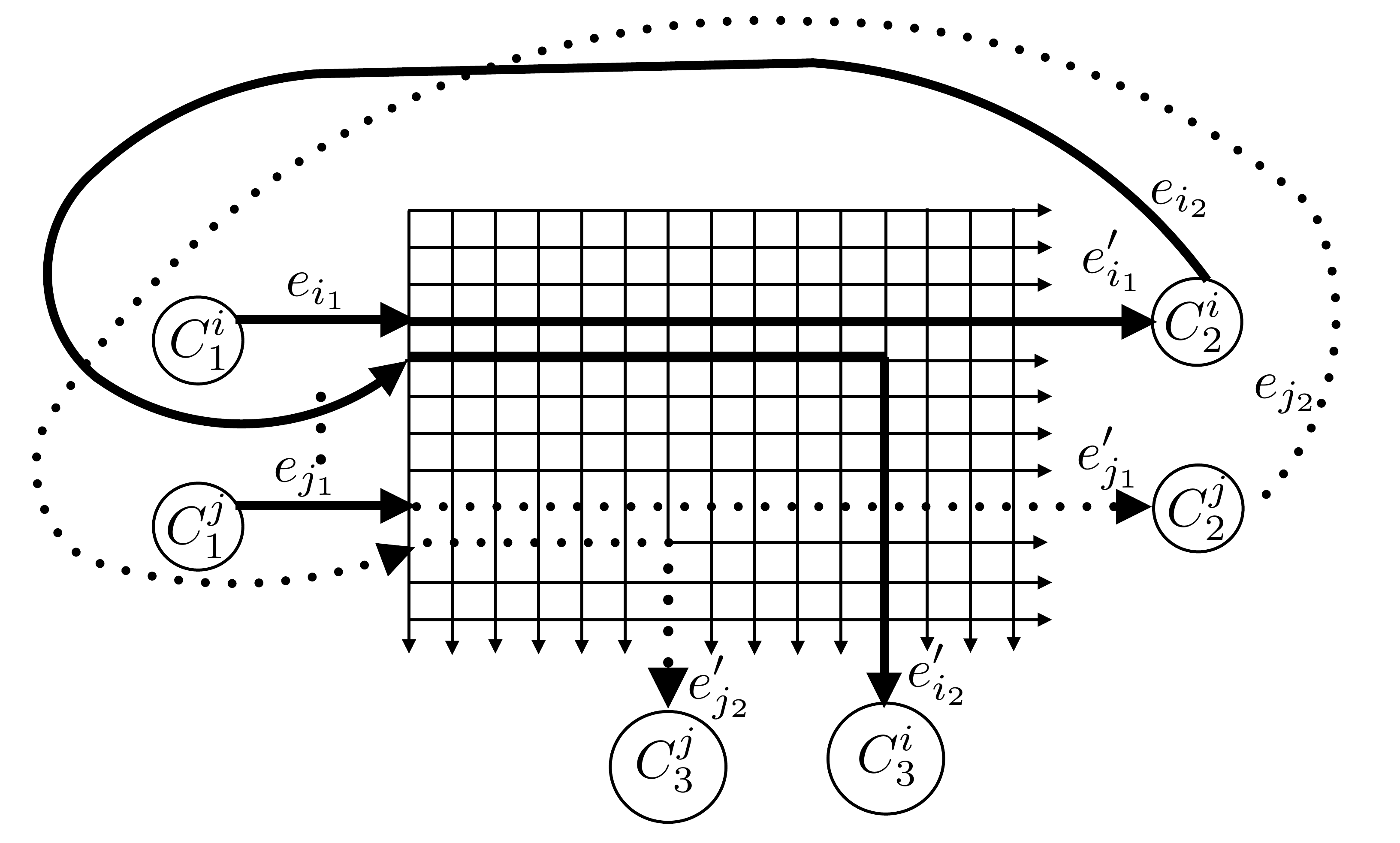}
		\end{center}
		\caption{Illustration of the construction in the proof of 
		Item~\ref{lem:3cycles} in Lemma~\ref{lem:biglemma}.}\label{fig:3cycles}
	\end{figure} 
For $i\in[k]$ attach $H_i$ to $A_{2k}$ (recall that $A_{2k}$ is an
acyclic grid of order $2k$) by adding edges
$e_{i_1}=(u,v_{2i-1,1})$, $e_{i_2}=(v_{2k,2i-1},y)$,
$e'_{i_1}=(v_{2i-1,2k},v)$, $e'_{i_2}=(x,v_{2i,1})$ to obtain a graph $A$. See
Figure~\ref{fig:3cycles} for illustrations. It is easy to see
that for any set of vertices $S\subseteq V(A)$ of size at most $f(2)$
we have $H\butterfly A - S$. In the rest we show that there are no two
distinct models of $H$ in $A$ as a butterfly minor.

Let $H'=\mu(H)$ be a minimal tree-like butterfly model of $H$ in $A$, i.e.~a 
	tree-like model such that no proper subgraph of $\mu(H)$ contains a model of 
	$H$. As $\mu$ is tree-like, every $\mu(v)$ is the union of 
	an in-branching 
	and an out-branching which only share their root $r_v$ and for every strongly connected subgraph 
	$X\subseteq H$ 
	the model $\mu(X)$ contains a maximal strongly connected 
	subgraph $\rho(X)$ 
	which contains all roots $r_v$ of $v\in X$.

Let $C_1^i,C_2^i,C_3^i$ be 
	the copies of $C_1,C_2,C_3$ 
	in $H_i$ 
	used to construct 
	$A$. 

	As $|C_2| < |C_3|$, $\rho(C_3)$ cannot be contained in $C_2^i$ for 
	any $1\leq i \leq 
	k$. Also for all $j\in [k]$, $\rho(C_3)$ cannot appear in any $X^j_1,\ldots,X^j_l$
        where $X^j_1, \dots, X^j_l$ 
	are the components of $H^j$ 
	such that $C^j_2$ is reachable from them (w.r.t. reachability
        of $C_2$ in $H$). 
        Therefore $\rho(C_3) \subseteq H_i - V(X^i_1 \cup \dots \cup
        X^i_l)$, for some $1\leq i \leq k$. On the other hand, $\rho(C_2) \subseteq
        \bigcup_{j\in [k]} H_j - V(Y^i_1 
	\cup \dots \cup Y^j_{l'}$, where $Y^j_1,\ldots,Y^j_{l'}$ are
        strong components
        of $H^j$ which are reachable from $C_2$. Note that the
        difference between the two cases is that model of $C_2$ can be obtained
        by going back and forth through arbitrarily many $C^j_2$'s as
        there are directed paths which connects them together.

Analogously we have $\rho(C_1) \subseteq
        \bigcup_{j\in [k]} H_j - V(C_j) \cup Y^i_1 
	\cup \dots \cup 
	Y^j_{l'}$.
There is a directed path $P_1$ in $H'$ which connects $\rho(C_1)$ to $\rho(C_2)$ and
there is a directed path $P_2$ in $H'$ which connects $\rho(C_2)$ to
$\rho(C_3)$. Considering the structures of
$\rho(C_1),\rho(C_2),\rho(C_3)$ as explained, $P_1,P_2$ will go through the acyclic grid in
$A$ (maybe they go through some of $C_2^i$'s as well) and they will cut
the acyclic grid into different regions. The path $P_2$ has a
subpath $P^{H'}_2$ which starts at some vertex $v_{i,1}$ and ends at
a vertex $v_{2k,j_2}$ for $2\le i \le 2k$ and $j_2\in
[2k]$. Also $P_1$ has a subpath $P^{H'}_1$ which starts at the
vertex $v_{i-1,1}$ and ends at the vertex $v_{j_1,2k}$ for some
$j_1\in [2k]$. Let $H''$ be another model of $H$ in $A$. 
Similar to $H'$ we have subpaths $P^{H''}_1,P^{H''}_2$. We have
$P^{H'}_1\cap P^{H''}_2 \neq \emptyset$ or $P^{H'}_2\cap P^{H''}_1 \neq \emptyset$.
Hence every two models of $H$ in $A$ will intersect.
\end{proof}

Proof of Theorem~\ref{thm:ep-vc} follows from Lemma~\ref{lem:biglemma}
and~\ref{lem:no-degree-4-vertex}.



\subsection{Positive Instance for Erd\H os-P\'osa Property in Vertex
  Cyclic Graphs}
We close the section by giving a positive result, i.e.~we provide a 
vertex-cyclic ultra-homogeneous digraph that has the Erd\H 
os-P\'osa 
property, but it is not strongly connected. 

\begin{theorem}\label{thm:two-random-cycles}
   Let $H$ be a digraph consisting of two disjoint cycles joined by a
  single edge. There is a function $h:\N\rightarrow\N$ such that
  for every integer $k$ and every graph $G$ either there are $k$
  distinct topological models of $H$ in $G$ or there is a set
  $S\subseteq V(G)$ such that $|S|\le h(|H|+k)$ and $H\not\topo
  G-S$. 

  Furthermore, for every $H$ and $k$ there is a polynomial-time
  algorithm which either finds $k$ distinct
  topological models of $H$ in $G$ or finds a set $S\subseteq G$ of
  vertices of size at
  most $h(|H|+k)$ which hits every topological model of~$H$ in~$G$.
\end{theorem}

In the rest of this section whenever we refer to $H$, it means the $H$
as stated in Theorem~\ref{thm:two-random-cycles}. Before providing
a proof for the theorem we need some lemmas and
definitions. An \emph{$l$-cycle} is a cycle of length at least $l$. 
An \emph{$l$-cluster} in a graph $G$ is a maximal subgraph of $G$
consisting of $l$-cycles such that every two of them intersect each
other. For a set 
$\mathcal{C}$ of
$l$-clusters in $G$ we write $G[\mathcal{C}]$ to denote
the subgraph of $G$ induced by the set of vertices occurring in 
an $l$-cluster in $\mathcal{C}$.

For three disjoint $l$-cycles $C_1,C_2,C_3 \subsetneq G$, the cycle $C_1$
is an \emph{$l$-transit} cycle of $C_2, C_3$ if 
there is a path $P$ in $G$ which connects
$C_2$ to $C_3$ and $P\cap C_1 \neq
\emptyset$. A cycle $C$ is an \emph{$l$-transit cycle} if it is 
an $l$-transit cycle for some pair $C_2, C_3$ of disjoint $l$-cycles.
A set $\mathcal{C}$ of $l$-clusters in $G$ is \emph{bipartite} 
if all $C\in \mathcal{C}$ are pairwise vertex disjoint and there is no $l$-transit cycle 
in $G[\mathcal{C}]$. A graph $G$ is \emph{$l$-cluster bipartite graph} if all of
its $l$-clusters together form a bipartite set of $l$-clusters. 
By Theorem~\ref{thm:EPPositive}, a single $l$-cycle has the Erd\H
os-P\'osa property, witnessed by some
function $f_1\colon \N\rightarrow \N$. In
   particular $f_1(2)$ means that in a given graph~$Z$ either there
   are two disjoint cycles of length at least $l$ or there is a set $S_1$ of
   size at most $f_1(2)$ such that there is no cycle of length at least $l$ in 
   $Z-S_1$.

The following lemma is required for some algorithmic aspects of
Theorem~\ref{thm:two-random-cycles}.

\begin{lemma}\label{lem:lclusterfinder}
There is an algorithm which for a given $l$-cluster bipartite graph $G$ of directed tree-width at most
$w$, finds all of its $l$-clusters in time
and space $|G|^{O(w+l)}$. 
\end{lemma}
\begin{proof}
We first observe that we can check whether a vertex $v\in V(G)$ lies
in an $l$-cycle or not, and if it is in some $l$-cycle find at least
one of those cycles, namely, its \emph{corresponding $l$-cycle}. To see this
for a vertex $v\in V(G)$ we guess $l-1$ distinct other vertices which together
with $v$ form a model of $l$-cycle. By
Theorem~\ref{thm:dtw:main-algo} we can check if they form an $l$-cycle
in time and space $|G|^{O(w+l)}$. So we can find corresponding
$l$-cycle of each vertex $v\in V(G)$. We put two vertices
$u,v\in V(G)$ in
one $l$-cluster if their corresponding $l$-cycles intersect. Recall
that in an $l$-cluster bipartite graph it is impossible
to have three $l$-cycles $C_1,C_2,C_3$ such that $C_1\cap C_2 \neq
\emptyset$ and $C_2\cap C_3\neq \emptyset$ and $C_1\cap C_3 =
\emptyset$, because then $C_2$ is an $l$-transit cycle.
\end{proof}

We use the following essential lemma in the rest of this section.

\begin{lemma}
\label{lem:bipartition}
Given an integer $k$ and $l$-cluster bipartite graph $G$. 
Either there are $k$ disjoint minors of $H$ in $G$ or there is a set $S_k
\subseteq V(G)$ such that it hits every model of $H$ in
$G$. Furthermore $|S_k| \le 2(k-1)\cdot f_1(2) +
(l-1)k(\max\{f_1(2),l\}) + k-1$.
\end{lemma}
\begin{proof}
We break the proof into three steps. First we either find $k$ disjoint
models of $H$ which have $2$ cycles of length at least $l$ or a set
$S_1\subseteq G$ of size
at most $k-1+(k-1)\cdot f_1(2)$ which hits every model of $H$ in $G$ which
has two $l$-cycles. We know every such model has its cycles in two
different $l$-clusters.

Let $\mathcal{C}$ be the set of all $l$-clusters in $G$.
If there is an element in $\mathcal{C}$ which does not
have a path to (or from) any other element of $\mathcal{C}$ in $G$
then it does not participate in any minor of $H$ as required in the
above so we can ignore them. We partition the
rest of $\mathcal{C}$ into two
partitions $\mathcal{C}_1,\mathcal{C}_2$ such that for any element in
$\mathcal{C}_1$ there is a path to some element in $\mathcal{C}_2$. As
non of the elements has an $l$-transit cycle, we always have this bi-partition.

We add a vertex $v_1$ to $G$ and for each $C'\in \mathcal{C}_1$ an edge from 
$v_1$ to a vertex in $C'$.
Similarly, add a vertex $v_2$ and for each $C''\in \mathcal{C}_2$ an edge from 
one vertex of $C''$ to  $v_2$.

By Menger's theorem, either there are $k$ disjoint paths from $v_1$ to
$v_2$ or there is a set of vertices of size at most $k-1$ which disconnects
$v_1$ from $v_2$. 

First we claim that if there are $k$ disjoint paths from $v_1$ to $v_2$,
then we have $k$ disjoint copies of $H$ as required. It is clear that any path
from $v_1$ to $v_2$ corresponds to a model of $H$ in $G$ with both cycles in 
$G[\mathcal{C}]$. On the
other hand, two models from two
disjoint paths may intersect only if they go through each others
components in $\mathcal{C}_1$ or $\mathcal{C}_2$. But this cannot happen, as otherwise we
have an $l$-transit cycle in $\mathcal{C}$.

Similarly, if there is a vertex set $W$ of size at most $k-1$ that disconnects $v_1$ and
$v_2$, then for every $v\in W \cap V(G[\mathcal{C}_1\cup\mathcal{C}_2])$ let
$S_v$ be the set of vertices of size at most $f_1(2)$ which hits every
cycle of length at least $l$ in the $l$-cluster that
$v$ belongs to. Let $S_1 = W\cup
\bigcup_{v\in W} S_v$. Then $S_1$ is a hitting 
set of every model of $H$ in $G$ which obtained from $2$ disjoint $l$-cycles. But 
size of $S_1$ is at most $k-1+(k-1)\cdot f_1(2)$ as claimed.

Now in the second step we consider each $l$-cluster in $G-S_1$. Each
$l$-cluster is strongly connected. Suppose there are $t$ disjoint $l$
clusters $\mathcal{C}_1,\ldots,\mathcal{C}_t$ such that for all
$i\in[t] : H\topo
\mathcal{C}_i$. If $t\ge k$ then we have $k$
disjoint models of $H$ in $G$. Otherwise for all $i\in [t]$ we can
choose a set $S'_i\subseteq V(\mathcal{C}_i)$ of size at most $f_1(2)$
vertices such that $\mathcal{C}_i - S'_i$ has no $l$-cycle. Let
$S_2=\bigcup_{i=1}^t S'_i$. We have $|S_2| \le (k-1) f_1(2)$. In
$G-S_1-S_2$ there is no model of $H$ which has both of its cycles in
one $l$-cluster.

In the third step we proceed on $G-S_1-S_2$. Take a set $\mathcal{C}'$
of all $l$-clusters in $G-S_1-S_2$. Take a set of \emph{corresponding small
cycles $\mathcal{C}''$} of
maximum size which consisting of disjoint cycles of size at least $s$ in
$G-S_1-S_2-\mathcal{C}'$. By our choice of $S_1,S_2,\mathcal{C}'$ it
is clear that every corresponding small cycle has length at most
$l-1$. Like
a first step, add a vertex $v_1$ to $G-S_1-S_2$
and for each $C'\in \mathcal{C}'$ an edge from 
$v_1$ to a vertex in $C'$.
Add a vertex $v_2$ and for each $C''\in \mathcal{C}''$ an edge from 
one vertex of $C''$ to  $v_2$. By Menger theorem either there are
$(l-1)k$ internally vertex disjoint
paths $\mathcal{P}$ from $v_1$ to $v_2$ or there is a hitting set of size at most
$k(l-1)-1$ which hits every path from $v_1$ to $v_2$. 

In the first
case we can find $k$ disjoint
models of $H$ as follows. For $P=\{v_1,u_1,\ldots,u_n,v_2\} \in
\mathcal{P}$, we say $u_1$ is the start point and $u_n$ is the end
point of the path $P$. We know that each path in $\mathcal{P}$ denotes a model of
$H$. Furthermore, by the first step (choice of $S_1$) start point of each
two paths are on two disjoint
$l$-cycles $c_1,c_2$ and there is no path between
$c_1,c_2$.

Each corresponding small cycle can route at most $l-1$ paths. We give
the following recursive algorithm to find a set $\mathcal{H}$ of at
least $k$ disjoint models of $H$ in
$G-S_1-S_2$. Take a path $P\in \mathcal{P}$ and let $c$ be its endpoint
corresponding small cycle. Suppose $P'_1,\ldots,P'_t\in \mathcal{P}$
intersecting $c$. We know that $t\le l-1$ as
the size of $c$ is at most $l-1$.
Put the corresponding model of $H$ w.r.t. $P$
in $\mathcal{H}$. Set $\mathcal{P}:=\mathcal{P}\setminus
\{P'_1,\ldots,P'_t\}$ and recurse. 

In each step, the algorithm finds a
model of $H$ which is disjoint from any other model which are already
in $\mathcal{H}$, so at the end $\mathcal{H}$ consists of disjoint
models of $H$. Furthermore in each step algorithm eliminates at most
$l-1$ paths from $\mathcal{P}$. So algorithm will run for at least $k$
steps, that follows $\mathcal{H}$ has at least $k$ disjoint
models of $H$.

If there is a hitting set $S$, then for
every $v\in S$ we create a set $S_v$ as follows. We set $S_v:=\{v\}$. If $v\in S\cap
\mathcal{C}'$ set the $S_v\subseteq V(G)$ of size at most
$f_1(2)$ which hits every $l$-cycle in strongly connected component of
$v$. If a vertex $v\in S\cap
\mathcal{C}''$ then $v\in c$ for some $c \subseteq
\mathcal{C}''$ and we set $S_v := V(c)$, in this case we have $|S_v|\le l-1$. Set
$S_3:=\bigcup_{v\in S} S_v$. We claim $S_3$ hits every model of $H$ in
$G-S_1-S_2$. 

Suppose there is a model $H'$ of $H$ in
$G-S_1-S_2-S_3$ consisting of cycles $c_1,c_2$ with a path from $c_1$
to $c_2$. In our construction, there is no path between $v_1,v_2$ by the choice of
$S_3$, so either the $c_1$ has no incoming edge from $v_1$ or the
$c_2$ has no edge to $v_2$. 

We claim either $c_1$ or $c_2$ does not exist so there is no such
$H'$ at all. Suppose $c_1$ exists. We know that
$c_2$ is not in any $l$-cluster of $G-S_1-S_2$ (recall the choice of
$S_2$). So $c_2$ is a cycle disjoint from any $l$-cluster and
therefore either is in $\mathcal{C}''$ or intersects $c'\in
\mathcal{C}''$. As $c_1$ exists, it means we did not take any vertex
from its $l$-cluster into $S_3$, so there is a path from $v_1$ to
$c_1$ and therefore to
$c_2$. In order to destroy connections from $v_1$ to $v_2$ we chose a
vertex $v\in c'$ by Menger
algorithm and therefore $V(c')\cap S_3 = V(c')$, but then $c_2\cap S_3 \neq
\emptyset$, so $c_2$ does not exist.

The size of $S_3$ is at most
$(l-1)(k-1)(\max\{f_1(2),l\})$. 
There is no model of $H$ in
$G-S_1-S_2-S_3$, we set $S_k=S_1\cup S_2\cup S_3$. The size of $S_k$
is at most $2(k-1)\cdot f_1(2) +
(l-1)k(\max\{f_1(2),l\}) + k-1$ as claimed.
\end{proof}

\begin{lemma}
\label{lem:main}
   	There is  a 
   function $f(k, 
   w)\colon \N\times\N\rightarrow \N$ such that for every $k, 
   w\geq 0$, every 
   digraph $G$ of directed 
   tree-width at most $w$ either contains $k$ disjoint 
   topological models of $H$ or a set of at most $f(k, w)$ 
   vertices 
   hitting every model of $H$. 
\end{lemma}
\begin{proof}
For the proof 
    of the lemma we 
    need a special form of directed tree decompositions. 
    A directed tree-decomposition $(T, \beta, \gamma)$ is 
    \emph{special}, if 
    \begin{enumerate}[nosep]
    \item for all $e = (s, t) \in E(T)$ the set 
        $\beta(T_t) := \bigcup_{t \preceq_T t'} \beta(t')$ 
        is a
        strong component of $G-\gamma(e)$ and 
        \label{def:dtw:1}
        \item $\bigcup_{t \prec_T t'}\beta(t') \cap 
        \bigcup_{e \sim t} \gamma(e)=\emptyset$ for every $t 
        \in V(T)$.\label{def:dtw:2}
    \end{enumerate} 
    It was shown in \cite{KreutzerO14} that every digraph of 
    directed tree width $w'$ has a special directed tree 
    decomposition of width at most $5w'+10$. 

    We set $f(0, w) = f(1, w) = 0$ and for $k>1$ we set $f(k,w) 
   := 5w+10+2f_1(2) + f(k-1,w) + 3|S_k|$, where $S_k$ is as provided
   in the Lemma~\ref{lem:bipartition}.

   Let $G$ be a digraph of directed tree-width at most $w$ and 
   let $(T, \beta, \gamma)$ be a directed tree-decomposition of 
   $G$ of width $w$.  For $t\in T$
   let $G_t:= G[\beta(T_t)]$. We prove the lemma by induction on $k$. 
   Clearly, for $k=0$ or $k=1$ there is nothing to show.
   So let $k>1$. If $H\not\topo G$, then again there is nothing 
   to show. Otherwise, let $t$ be a node in $T$ of minimal 
   height such that $H \topo G_t$. By definition of special 
   directed 
   tree-decompositions, for every successor $c$ of $t$ the
   digraph $G_c$ is strongly connected.  So if $c$ is a 
   successor of $t$ then 
   $G_c$ does not contain 
   two
   disjoint cycles of length at least $l$ as otherwise $H \topo 
   G_c$ contradicting the choice of $t$. So there is a hitting 
   set $S_c$ of size at most $f_1(2)$
   such that $G_c-S_c$ has no cycle of length at least
   $l$.

    Let $\sqsubseteq$ be a linearisation of the topological 
    order of the children of $t$. Let $F=
   (c_1,\ldots,c_m)$ be the tuple of children of $t$  satisfying 
   the following conditions:

    \begin{enumerate}[nosep]
    \item $H\topo  F(t) := \Gamma(t) \cup\bigcup_{c\in F} G_c$.
    \item Subject to 1, $F$ is the
      lexicographically smallest tuple w.r.t. $\sqsubseteq$. 
    \end{enumerate}

    It is easy to see that there are no $3$ distinct nodes 
    $c^1,c^2,c^3 \in
   F$  such that there is a cycle of length at least $l$ in
   $G_{c^1} - \Gamma(t), G_{c^2} - \Gamma(t),G_{c^3} - \Gamma(t)$, as
   otherwise we could choose a smaller set $F$ satisfying the 
   conditions, contradicting the fact that $F$ satisfies the 
   second condition.

    So suppose there are at most two nodes $c^1,c^2$ in
    $F$ containing  a
    cycle of length at least $l$ in $G_{c^1} - \Gamma(t)$ and 
    $G_{c^2} -
    \Gamma(t)$, respectively. Let $S(t) := \Gamma(t) \cup 
    S_{c^1} 
    \cup S_{c^2}$. By construction, $S(t)$ hits every cycle of 
    length at least $l$ in $F(t)$.  Hence, in $G_0 := F(t) - 
    S(t)$
    there is no minor of $H$ but there is a minor of $H$ in 
    $F(t)$. 

    If $G-F(t)$ contains $k-1$ disjoint 
   topological models of $H$ then this implies that $G$ has 
   $k$ disjoint models of $H$ and we are done. 
   Otherwise, by induction hypothesis, there is a set 
   $S\subseteq V(G-F(t))$ of order at most
   $f(k-1,w)$ such that $H\not\topo G-F(t)-S$. 
   
    Note that every model of $H$ 
    in $G - S - S(t)$ must contain vertices of $G_0$ and also
   vertices of $G - S - S(t)-F(t)$.
    Let $G_1 := (G - S - S(t)-F(t)) \cap G[\beta(T_t)]$ and $G_2 
    := (G - S - S(t)-F(t))- G_1$.

    In the rest of the proof, we will first construct a hitting 
    set for every model of $H$ in $G_{0,1}:= G_0\cup G_1$,
then construct a hitting set of models of $H$ which have both of 
their cycles
in $G_2$ connected by a path containing vertices of $G_1\cup 
G_0$ and finally find a hitting set of models of $H$ which have 
one cycle
in $G_2$ and the other in $G_1$. 
In any of the three cases, if we fail to find the required
hitting set, we output $k$ disjoint models of $H$. As no other 
choice
of any model of $H$ remains, we are done with the proof.

By construction there is no cycle of length at least $l$
in $G_0$. Also by construction there is no minor of $H$ in each of
$G_0,G_1,G_2$. 
By Lemma~\ref{lem:scc-dtw} there is no path $P$ in $G[G_0\cup G_1 \cup G_2]$ with start and
end point in $G_1$ such that $P \cap G_2 \neq \emptyset$. 

If $H\topo G_{0,1}$, then let 
$\mathcal{C}_{G_{0,1}}$ be the set of all
$l$-clusters in $G_{0,1}$. As $G_0$ does not contain any cycle 
of length at least $l$, the clusters in $\mathcal{C}_{G_{0,1}}$ 
are all contained in $G_1$. Furthermore, no two distinct clusters can 
share a vertex  as otherwise there would be a minor of $H$ in
$G_1$. Finally, in $G_{0,1}$ there cannot be an $l$-transit 
cycle as otherwise 
the choice of $F$ would not have been minimal w.r.t. $\sqsubseteq$. For, 
suppose there was an 
$l$-transit cycle $C_1$ in $G_{0,1}$, i.e.~there are $l$-cycles $C_1, C_2, C_3$ 
in $G_{0,1}$ and a path from $C_2$ to $C_3$ containing a vertex of $C_1$.
As $G_0$ does not contain any $l$-cycle, $C_1, C_2, C_3$ are all in $G_1$. But 
as $G_1$ does not contain $H$ as a topological minor, the subpath of $P$ from 
$C_2$ to $C_1$ and also the subpath of $P$ from $C_1$ to $C_3$ must contain a 
vertex of $G_0$. But this implies that $F$ does not satisfy the second
condition.
 
So $G_{0,1}$ is a $l$-cluster bipartite graph. The following
Lemma~\ref{lem:bipartition} shows
that in any $l$-cluster bipartite graph either there are $k$ disjoint models
of $H$ or a small set of vertices are a hitting set for all models of $H$. So in $G_{0,1}$
either we find $k$ disjoint minors of $H$ or there is a set
$S_{G_{0,1}}$ that hits every minor of $H$ in $G_{0,1}$. In the first
case we are done, so suppose we have the set $S_{G_{0,1}}$.
Now we have to consider all $l$-clusters in $G_2$.
\begin{Claim}
\label{clm:no3cycleG2}
There are no $3$ cycles
$c_1,c_2,c_3$ of length at least $l$
in $G_2$ such that there is a
path $P_1$ from $c_1$ to $c_2$ and a path $P_2$ from $c_2$ to $c_3$ in
$G-S-S(t)$.
\end{Claim}
\begin{ClaimProof}
We know that there is no minor of $H$ in $G_2$. If there are $3$
cycles as stated in the claim, then both paths $P_1, P_2$ must contain a vertex 
of
$G_0 \cup G_1$. But then there is a path between two vertices of 
$G[\beta(T_t)]-\Gamma(t)$ which
does not go through $\Gamma(t)$ but intersects $G-G[\beta(T_t)]$. But this is  
impossible by
Lemma~\ref{lem:scc-dtw}.
\end{ClaimProof}

By Claim~\ref{clm:no3cycleG2} and Lemma~\ref{lem:bipartition} and the
fact that all $l$-clusters in $G_2$ are vertex disjoint, 
either there are $k$ disjoint
models of $H$ in $G - S - S(t) -S_{G_{0,1}}$ such that both of their cycles are in $G_2$ or
there is a set of vertices $S_{G_2}\subseteq V(G - S - S(t))$ such that there is no
model of $H$  in $G' = G - S - S(t) -S_{G_{0,1}} - S_{G_2}$ with both of its cycles in
$G_2 - S_{G_{0,1}} - S_{G_2}$. In the first case we are done. So suppose we
have $S_{G_2}$ as in Lemma~\ref{lem:bipartition}.

Any model of $H$ in $G'$ must map one cycle of $H$ to  
   $G_{0,1} - S_{G_{0,1}} - S_{G_2}$ and the other to the $G_2 - S_{G_{0,1}} - S_{G_2}$.

Let $\mathcal{C}$ be the set of clusters in $G - S - S(t) - S_{G_{0,1}} -
S_{G_2}$. All $l$-clusters in $\mathcal{C}$ are
pairwise vertex disjoint.

There is no path between two clusters $c_1,c_2 \in G[\mathcal{C}]\cap
G_1 - S_{G_{0,1}} - S_{G_2}$ because there is no such path fully in $G_{0,1}$ 
and it
cannot go through a vertex $v \in G_2- S_{G_{0,1}} -
S_{G_2}$ by Lemma~\ref{lem:scc-dtw}. There is no cluster $c$ in
$G_{0,1}-S_{G_{0,1}} - S_{G_2}$ such that it has a path to a cluster $c'\in
G_2-S_{G_{0,1}} - S_{G_2}$ and a path from a cluster $c'' \in G_2 -
S_{G_{0,1}} - S_{G_2}$, as otherwise there is a path between $c',c''$ in
$G'$. So $\mathcal{C}$ is a bipartition and by
Lemma~\ref{lem:bipartition} either there are $k$ disjoint minors of
$H$ in $G'$ or
there is a hitting set $S_{G'}$ in $G'$.
So either there are $k$ disjoint minors of $H$ in $G$ or a set $S_G
= (G-G') \cup S_{G'}$ of size at most
$w+2f_1(2) + f(k-1,w) + 3|S_k|$ which hits every minors of $H$ in $G$.
\end{proof}

\begin{proof}[Proof of Theorem~\ref{thm:two-random-cycles}]
   Let $H$ be as in the statement of the theorem with two cycles
   $C_1,C_2$. Let $l,s$ be the
   length of $C_1,C_2$ resp. such that $l\ge s$. W.l.o.g suppose there
   is an edge from $C_1$ to $C_2$. 
	
	Let $g \st \N \rightarrow \N$ be the function as defined in 
	Theorem~\ref{thm:grid}. We claim that $h\st \N\rightarrow \N$ 
	defined by $h(k) := f(k, g((k+1)\cdot l))$ witnesses the Erd\H os-P\'osa 
	property of $H$. To see this, let $G$ be any digraph and let $k \geq 1$. If 
	the directed tree-width of $G$ is at least $g((k+1)\cdot l)$, then by 
	Theorem~\ref{thm:grid}, $G$ contains the 
	cylindrical wall 
	$W_{(k+1)\cdot l}$ of order $(k+1)\cdot l$ as topological minor, which 
	contains $k$ disjoint copies of $H$ as topological minor. Otherwise, i.e.~if 
	the directed tree-width of $G$ is $<g((k+1)\cdot l)$, then by
        Lemma~\ref{lem:main}, $G$ contains $k$ disjoint topological
        models of $H$ or a set $S$ of 
	at most $f(k, g((k+1)\cdot l))$ vertices such that $H\not\topo
        G-S$.
\end{proof}


\section{Conclusion}

In this paper we studied the generalised Erd\H os-P\'osa property for directed 
graphs with respect to topological and butterfly minors. We provided an exact 
generalisation of Robertson and Seymour's classification of undirected graphs 
with the Erd\H os-P\'osa property to strongly connected digraphs. Furthermore, 
for the natural and much larger class of vertex-cyclic digraphs we obtained an 
almost exact characterisation. We also provide a novel approach to
prove the  Erd\H os-P\'osa property holds in the special case of vertex cyclic
graphs.

We believe that the techniques developed here 
will provide the tools to give a complete characterisation of the vertex-cyclic 
digraphs with the Erd\H os-P\'osa property but we leave this to future research.


\bibliography{dep}


\end{document}